\documentclass[onecolumn, 12pt]{IEEEtran}

\usepackage{amsmath, amsthm, amssymb, amsfonts, graphicx, graphs, fancyhdr, psfrag, enumerate,float,subcaption}

\newcommand\norm[1]{\left\lVert#1\right\rVert}

\newcommand\bbr{\mathbb{R}}

\renewcommand\v[1]{{\mathbf{#1}}}

\newcommand\mc[1]{\mathcal{#1}}

\DeclareMathOperator*{\supp}{supp}

\newcommand\vlambda{{\boldsymbol{\lambda}}}
\newcommand\vomega{{\boldsymbol{\omega}}}
\newcommand\vpsi{{\boldsymbol{\psi}}}
\newcommand\veta{{\boldsymbol{\eta}}}

\newtheorem{theorem}{Theorem}[section]
\newtheorem{lemma}[theorem]{Lemma}
\newtheorem{proposition}[theorem]{Proposition}

\theoremstyle{definition}
\newtheorem{defn}[theorem]{Definition}

\begin{document}

\title{Using Pseudocodewords to \\ Transmit Information}

\author{Nathan Axvig\thanks{N. Axvig is with the Department
  of Mathematics, Concordia College, Moorhead, MN, 56562, USA (email:  ndaxvig@cord.edu).}
	\thanks{This paper was presented in part at the 2011 Fall Central Section Meeting of the American Mathematical Society (Lincoln,
NE; October 2011) and at the 2012 Joint Mathematics Meetings (Boston, MA; January 2012).}
	\thanks{This work has been submitted to the IEEE for possible publication. Copyright may be transferred without notice, after which this version may no longer be accessible.}
}

\maketitle


\begin{abstract}
The linear programming decoder will occasionally output fractional-valued sequences that do not correspond to binary codewords -- such outputs are termed \emph{nontrivial pseudocodewords}.  Feldman et al. have demonstrated that it is precisely the presence of nontrivial pseudocodewords that prevents the linear programming decoder from attaining maximum-likelihood performance.  The purpose of this paper is to cast a positive light onto these nontrivial pseudocodewords by turning them into beneficial additions to our codebooks.  Specifically, we develop a new modulation scheme, termed \emph{insphere modulation}, that is capable of reliably transmitting both codewords and pseudocodewords.  The resulting non-binary, non-linear codebooks have higher spectral efficiencies than the binary linear codes from which they are derived and in some cases have the same or slightly better block error rates.  In deriving our new modulation scheme we present an algorithm capable of computing the \emph{insphere} of a polyhedral cone -- which loosely speaking is the largest sphere contained within the cone.  This result may be of independent mathematical interest.
\end{abstract}

\section{Introduction}\label{sec:intro}

When simulating a block code, one often repeats the following procedure:  generate a codeword, transmit it over a noisy channel, estimate a codeword that best explains the received vector, and finally compare the estimated codeword to the original.  When the estimated codeword and the transmitted codeword are identical, we record an overall success for the coding scheme.  In the case where the estimated codeword is not equal to the original codeword, we declare a block error.  This paper explores a third option:  what if the estimated codeword is not a codeword at all?  What, if anything, can be gained from this?

The third option above is a very real possibility when using certain ``modern'' decoding algorithms such as min-sum and linear programming decoding.  It has been observed (see, e.g.,~\cite{aveMinSum, fwk05}) that occasionally these decoders will output sequences that are not codewords at all -- such non-codeword outputs are termed \emph{nontrivial pseudocodewords}.  For an example, consider transmission of a codeword from the binary linear code given by the null space of the following parity-check matrix over the additive white Gaussian noise (AWGN) channel:
\[
H = \footnotesize{\begin{bmatrix}
1 & 1 & 0 & 1 & 0 & 0 & 0 & 0 & 0 \\
0 & 1 & 1 & 0 & 1 & 0 & 0 & 0 & 0 \\
1 & 0 & 1 & 0 & 0 & 1 & 0 & 0  & 0\\
0 & 0 & 0 & 1 & 0 & 0 & 1 & 0 & 1 \\
0 & 0 & 0 & 0 & 1 & 0 & 1 & 1 & 0\\
0 & 0 & 0 & 0 & 0 & 1 & 0 & 1 & 1 \\
\end{bmatrix}}.
\]
If the received vector yields the vector of log-likelihood ratios $\begin{bmatrix} 0, 0, 0, -1, 1, 1, 0, 0, 0 \end{bmatrix}^T$, then the linear programming (LP) decoder will return $\begin{bmatrix} \frac{1}{2}, \frac{1}{2}, \frac{1}{2}, 1, 0, 0, \frac{1}{2}, \frac{1}{2}, \frac{1}{2}\end{bmatrix}^T$ as an output. This output vector is clearly not a binary codeword since it has fractional values.  In such a case Feldman et al. \cite{fwk05} suggest that we declare a block error and move on with our lives.  

We henceforth adopt the mindset that any potential output of a decoder can be used to transmit information across a channel.  To this end we construct a new modulation scheme, termed \emph{insphere modulation}, that is capable of reliably transmitting pseudocodewords in conjunction with codewords.  Employing insphere modulation therefore allows for larger codebooks of the same length as the original binary code, and as a result we see modest gains in both spectral efficiency and information rate.  Moreover, we demonstrate that in some circumstances these gains in throughput can be made without sacrificing error-correcting performance.  In order to compute inspheres for transmitting pseudocodewords, we develop an algorithm capable of computing the insphere of a polyhedral cone to any degree of accuracy.  While we presently treat this result as a means to an end, it may be of independent mathematical interest.

The remainder of this paper is organized as follows.  In Section~\ref{sec:background} we review linear programming decoding and some key results from linear optimization.  Section~\ref{sec:bestVector} is dedicated to determining which  vectors are best suited to transmitting linear programming pseudocodewords across the AWGN channel, and Section~\ref{sec:insphere} derives an algorithm for computing the best of these vectors, which in our case turns out to be the center of the insphere of an appropriately defined polyhedral cone.  Section~\ref{sec:csym} shows how the results of Section~\ref{sec:bestVector} can be expanded to greatly increase the size of our codebook.  Section~\ref{sec:example} gives a detailed example, Section~\ref{sec:performance} contains simulation results, and we conclude in Section~{\ref{sec:conclusion}.

\section{Background}\label{sec:background}

To produce vectors suitable for reliably transmitting pseudocodewords, we must go ``under the hood'' of the linear programming decoder and its underlying algebra and geometry.  The material in this section is largely adapted from~\cite{bt}, though it can also be found in most introductory texts on linear optimization.  Notationally, we write $\v x^T$ for the transpose of $\v x$ and $\norm{\v x}$ for the usual Euclidean norm of $\v x$.  Unless  stated otherwise, all vectors are assumed to be column vectors.

\subsection{Background on Linear Programming}\label{subsec:backgroundOnLinearProgramming}

A \emph{polyhedron} in $\bbr^n$ is the set of all points satisfying a finite set of linear equality and inequality constraints.  Since an equality constraint of the form $\v a^T \v x = b$ can be replaced by the two inequalities $\v a^T \v x \leq b$ and $\v a^T \v x \geq b $, we can equivalently define a polyhedron as the set of all points satisfying a finite number of linear inequality constraints.  Moreover, by multiplying constraints by $-1$ we may reformulate any description of a polyhedron in various compact forms involving matrix products such as $A_1 \v x \leq \v b_1 $ or $A_2 \v x \geq \v b_2$.  

If a polyhedron is bounded we refer to it as a \emph{polytope}.  On the other hand, if a polyhedron is closed under multiplication by non-negative scalars we call it a \emph{polyhedral cone}.  It is an easy exercise to see that polyhedra are convex sets, i.e., given two points $\v x$ and $\v y$ in a polyhedron $\mc M$ and a scalar $\lambda \in [0,1]$, the convex combination $\lambda \v x + (1-\lambda) \v y$ is also in $\mc M$.  If a point $\vomega$ in a polyhedron $\mc M$ cannot be written as a convex combination of two different points in $\mc M$, we say that $\vomega$ is an \emph{extreme point} or a \emph{vertex} of the polyhedron. Given a point $\v y \in \bbr^n$  and a constraint of the form $\v a^T \v x \leq b, \v a^T  \v x  \geq b, \text{ or } \v a^T  \v x  = b $, we say that the constraint is \emph{active} or \emph{binding} at $\v y$ provided that the constraint is met with equality, i.e., if $\v a^T \v y = b$.

Given a finite set of vectors $\mc V = \{ \v v_1, \v v_2, \dots, \v v_\ell\} \subseteq \bbr^n$ and a set of non-negative scalars $\{w_1, w_2, \dots, w_\ell\}$, we say that $ \sum_{i = 1}^\ell w_i \v v_i$ is a \emph{conic combination} of the vectors in $\mc V$. The polyhedral cone generated by $\mc V$ is then the set of all conic combinations of vectors in $\mc V$.  It is initially unclear that $\mc K(\mc V)$ is a polyhedron since it is not explicitly described as the set of all vectors satisfying a finite set of equality and inequality constraints.  A classical approach to affirming this fact is \emph{Fourier-Motzkin elimination} (see, e.g.,~\cite{bt}), a technique that produces a complete set of linear equality and inequality constraints for a projection of a polyhedron onto a subset of its coordinates.

For our purposes we define an optimization problem to be the task of maximizing or minimizing an \emph{objective function} over all points in a \emph{feasible set}.  A \emph{linear program (LP)} is simply an optimization problem for which the objective function is linear and whose feasible set is a polyhedron.  It is well known that if an optimal solution to a linear program exists, then there is at least one extreme point of the underlying polyhedron that achieves this optimal value.  Conversely, given an extreme point $\vomega$ of a polyhedron $\mc M$ there must exist some linear function $\v c$ so that $\vomega$ is the unique minimizer of $\v c^T \v x$ over all $\v x \in \mc M$ (see, e.g., pages 46-52 of~\cite{bt}).    

\subsection{Background on Linear Programming Decoding}\label{subsec:lpDecodingBackground}

In~\cite{fwk05} the \emph{linear programming decoder} is introduced.  This decoder operates by solving a linear programming relaxation of the maximum-likelihood (ML) decoding problem.  In particular, the decoder outputs a solution to the following linear program:
\begin{equation*}
\begin{array}{rll}
\text{minimize} & \vlambda^T \v x  \\
\text{subject to} & \v x \in \mc P(H)
\end{array}
\end{equation*}
where $\vlambda$ is the vector of log-likelihood ratios determined from the channel output and $\mc P(H) \subseteq \bbr^n$ is a polytope whose constraints are determined by the parity-check matrix $H$ chosen to represent the code.

\begin{defn}[\cite{fwk05}]\label{defn:fundamentalpolytope}
Let $C$ be a code presented by an $m \times n$ parity-check matrix $H = (h_{j,i})$.  For any row $\v h_j$ of $H$, define $N(\v h_j)$ to be the set of all indices $ i \in \{1, 2, \dots, n\}$ with $h_{j, i} = 1$.  The \emph{fundamental polytope\footnote{In~\cite{fwk05}, this is referred to as the \emph{projected} polytope.}} $\mc P = \mc P(H)$ is the set of all vectors $\v x \in \bbr^n$ satisfying the following constraints:
\begin{itemize}
\item $0 \leq x_i \leq 1$ for all $i = 1, 2, \dots, n$ and
\vspace{.1cm}
\item $\displaystyle \sum_{i \in S} x_i + \sum_{i' \in N(\v h_j)\setminus S}(1- x_{i'}) \leq |N(\v h_j)| - 1$ for all pairs $(j, S)$, where  $j \in \{1, 2, \dots, m\}$ and $S$ is a subset of $N(\v h_j)$ with odd cardinality.
\end{itemize}
\end{defn}

Since the fundamental polytope $\mc P(H)$ is bounded the linear programming decoder will always have an optimal solution, and by the preceding discussion we can assume that the linear programming decoder always returns an extreme point of $\mc P(H)$ as its output. As in~\cite{AMPTPW3}, we define a \emph{linear programming (LP) pseudocodeword} to be any extreme point of $\mc P(H)$. In~\cite{fwk05}, it is shown that all codewords are extreme points of $\mc P(H)$ and conversely that all integer-valued vectors within $\mc P(H)$ must be codewords.  Again following~\cite{AMPTPW3}, we therefore refer to codewords as \emph{trivial linear programming (LP) pseudocodewords} and any extreme points of $\mc P(H)$ with fractional values as  \emph{nontrivial linear programming (LP) pseudocodewords}.  Since this paper is  concerned only with the LP decoder, we will often refer to these simply as trivial and nontrivial pseudocodewords.

In~\cite{fwk05}, it is shown that if the linear programming decoder outputs an integer-valued codeword then this codeword must be a maximum-likelihood codeword.  Since the set of codewords and the set of of integer-valued points in $\mc P(H)$ are identical, this implies that the presence of nontrivial LP pseudocodewords is precisely what prevents the LP decoder from achieving ML performance.  Much research has therefore been devoted to examining the structure and properties of nontrivial pseudocodewords (see, e.g.,~\cite{AMPTPW3,KelSri3,kelley-sridhara,KLVW2,SmarVon07,VonKoe06,XiaFu}).  The authors of~\cite{dublin} go so far as to determine whether for a given code $\mc C$ there exists a parity-check matrix $H$ for which $P(H)$ contains no nontrivial pseudocodewords -- for most binary linear codes, the answer is ``no''~\cite{dublin}.  Nontrivial pseudocodewords are ubiquitous, and yet for certain received vectors they \emph{are} optimal solutions to the LP decoding problem.  Taking the saying ``If you can't beat them, join them'' to heart, we seek a method of incorporating nontrivial pseudocodewords into our codebooks so as to increase information rates and spectral efficiency without sacrificing error-correcting performance.

\section{How To Transmit a Pseudocodeword}\label{sec:bestVector}

When transmitting binary values over the additive white Gaussian noise channel it is common to employ binary phase shift keying (BPSK), which without loss of generality sends a binary 0 to $+1$ and a binary $1$ to $-1$. As a first (and very naive) attempt at transmitting nontrivial pseudocodewords, we simply extended this modulation mapping via the affine map $x \mapsto 1-2x$, added Gaussian noise, and applied LP decoding to the resulting vector.  The performance of this modulation scheme was terrible -- the block error rate was consistently close to 1 at all decibel levels. On investigation, we discovered that even when \emph{no noise} was added to our modulated vector the LP decoder was unable to recover the original pseudocodeword.  We concluded that to have any hope of transmitting a pseudocodeword a fundamentally new modulation scheme would be required.

The key observation in deriving our new modulation scheme for the transmission of linear programming pseudocodewords over the additive white Gaussian noise channel is that every pseudocodeword, trivial or nontrivial, is an extreme point of the fundamental polytope $\mc P(H)$.  By the discussion at the end of Section~\ref{subsec:backgroundOnLinearProgramming}, given a pseudocodeword $\vomega \in \mc P$ there must exist a vector of log-likelihood ratios $\vlambda$ so that $\vomega$ is the \emph{unique} solution to the following optimization problem:  minimize $\vlambda^T \v x$ over all $\v x \in \mc P$.  We therefore aim to produce a modulated vector that is equal to such a vector of log-likelihood ratios, since this will ensure that when $\vomega$ is transmitted in a low-noise scenario the LP decoder will recover the pseudocodeword $\vomega$.  Since on the AWGN channel the vector of log-likelihood ratios is proportional to the received vector, we can simply transmit the vector $\vlambda$ to accomplish this.  For this reason we will conflate the terms ``received vector'' and ``vector of log-likelihood ratios.''

As we prove in Proposition~\ref{prop:coniccombo} there are usually infinitely many choices for $\vlambda$.  It is our task to find a choice that minimizes the probability of block error.

\begin{proposition}\label{prop:coniccombo}
Let $\mc M = \{\v x \, | \, A\v x \geq \v b \}$ be a polyhedron in $\bbr^n$, let $\v x^\ast \in \mc M$ be given, and define $S$ to be the set of all indices $i$ where $\v a_i^T \v x^\ast = b_i$ - i.e., the index set for all constraints that are active at $\v x^\ast$. Consider the problem of minimizing $\v c^T\v x$ over all $\v x \in \mc M$: the vector $\v x^\ast$ is an optimal solution to this problem if and only if $\v c^T$ is a conic combination of the rows of $A$ indexed by $S$.
\end{proposition}

\begin{proof}
Define $A_S$ to be the matrix obtained by only considering those rows of $A$ indexed by $S$ and suppose that $\v c^T$ is not a conic combination of the rows in $A_S$.  This means that there is no vector $\v w$ satisfying both $\v c^T = \v w^T A_S$ and $\v w \geq \v 0$.  By Farkas' Lemma (see, e.g.,~\cite{bt}), there must exist a vector $\v d$ with $A_S \v d \geq 0$ and $\v c^T \v d < 0$.  Consider now the vector $\v x^\ast + \theta \v d$.  Since $\v a_i^T \v x^\ast > b_i$ for all $i \not \in S$, if we choose $\theta > 0$ to be sufficiently small then $\v x^\ast + \theta \v d$ will satisfy $A(\v x^\ast + \theta \v d) \geq \v b$ and thus be an element of $\mc M$.  Computing the objective value of this new vector in $\mc M$, we obtain
\[
\v c^T (\v x^\ast + \theta \v d) = \v c^T \v x^\ast + \theta \v c^T \v d < \v c^T \v x^\ast,
\] 
which implies that $\v x^\ast$ cannot be an optimal solution to the minimization problem.

For the other direction, we use a modification of the proof that  any basic feasible solution of a linear program is also a vertex (see, e.g., the discussion on page 52 of~\cite{bt}).  We now reproduce this proof with appropriate modifications for the sake of completeness.  Suppose that $\v c^T = \sum_{i \in S} w_i \v a^T_i$ for some scalars $w_i$ with $w_i \geq 0$ for all $i \in S$.  If $\v x$ is an arbitrary point in $\mc M$, we have $\v c^T \v x = \sum_{i \in S} w_i \v a_i^T \v x \geq \sum_{i \in S} w_i b_i$.  The number $\sum_{i \in S} w_i b_i$ is therefore a lower bound on the objective value for any point in $\mc M$.  But by the definition of $S$, we have that $\v a_i^T \v x = b_i$ for all $i \in S$.  Thus, $\v c^T \v x^\ast = \sum_{i \in S} w_i \v a_i^T \v x^\ast = \sum_{i \in S} w_i b_i$. The vector $\v x^\ast$ attains this bound and therefore must be an optimal solution to the problem of minimizing $\v c^T \v x$ subject to $A \v x \geq \v b$.

\end{proof}

To apply Proposition~\ref{prop:coniccombo} to our coding problem, we make the following definition.

\begin{defn}
Let $\mc C$ be a code presented by the parity-check matrix $H$, let $\mc P(H)$ be the associated fundamental polytope, and assume that $\mc P(H)$ has been written as $\mc P(H) = \{\v x \, | \, A\v x \geq \v b \}$ for an appropriate matrix $A$.  For any linear programming pseudocodeword $\vomega$, the \emph{recovery cone $\mc K_\vomega(H) = \mc K_\vomega$} of $\vomega$ is the cone generated by all row vectors of $A$ for which the corresponding constraint of $\mc P(H)$ is binding at $\vomega$.
\end{defn}

Let $t(\vomega)$ be a point in the recovery cone $\mc K_\vomega$ and let $\veta$ represent the additive Gaussian noise introduced by the channel.   By Proposition~\ref{prop:coniccombo}, the nontrivial pseudocodeword $\vomega$ can be correctly recovered if and only if the received vector $t(\vomega) + \veta$ lies within $\mc K_\vomega$.  If $r$ is the radius of the largest sphere centered at $t(\vomega)$ yet contained completely within $K_\vomega$, then the probability of correct decoding is bounded below by the probability that $|\veta| \leq r$.  

A trivial way of driving the probability of correct decoding to 1 is simply to scale $t(\vomega)$ by a large positive constant, which is tantamount to transmitting with more energy.  It is more interesting to place an energy constraint on the transmitted vector and then search for a vector that maximizes the radius of the largest sphere centered at that vector and contained within $\mc K_\vomega$. Stated more formally, we seek to compute the \emph{insphere} of  $\mc K_\vomega$.  This matter is treated separately in Section~\ref{sec:insphere}.  

Our approach to computing the insphere of a polyhedral cone depends on having access to an explicit description of the cone in terms of linearly inequality constraints.  While Fourier-Motzkin elimination (see, e.g.,~\cite{bt}) can produce such a description, it is terribly inefficient:  for codes of small block length ($n \leq 20$) Fourier-Motzkin elimination produces the constraint matrix for $\mc K_\vomega$ in a reasonable amount of time (a few minutes to a few hours on a desktop computer).  For longer codes, however, we approximate $\mc K_\vomega$ by examining a larger cone that contains it.

\begin{defn}\label{defn:Lvomega}
Let $\mc C$ be a code presented by the parity-check matrix $H$, let $\vomega$ be a linear programming pseudocodeword, and let $\Psi = \{\vpsi_1, \vpsi_2, \dots, \vpsi_r\}$ be a set of linear programming pseudocodewords.  The \emph{approximation cone $\mc R_{\vomega, \Psi}(H) = \mc R_{\vomega, \Psi}$} is defined to be the set of all vectors $\v x$ satisfying $R_{\vomega, \Psi} \v x \geq 0$, where the $i$th row of $R_{\vomega, \Psi}$ is given by $(\vpsi_i - \vomega)^T$.
\end{defn}

\begin{proposition}\label{prop:LcontainsK}
Let $\mc C$ be a code presented by the parity-check matrix $H$.  If $\vomega$ is a linear programming pseudocodeword and $\Psi = \{\vpsi_1, \vpsi_2, \dots, \vpsi_r\}$ is a set of linear programming pseudocodewords, then the recovery cone $\mc K_\vomega$ is a subset of the approximation cone $\mc R_{\vomega,\Psi}$.
\end{proposition}

\begin{proof}
The inequality $(\vpsi_i - \vomega)^T\v y \geq 0$ describes all cost functions $\v y$ of the linear programming decoder for which $\vomega$ has an objective value that is no larger than that of $\vpsi_i$.  By Proposition~\ref{prop:coniccombo}, it follows that every vector $\v x \in \mc K_\vomega$ satisfies $(\vpsi_i-\vomega)^T\v x \geq 0$.  Intersecting over all $i$ we obtain $\mc K_\vomega \subseteq \mc R_{\vomega, \Psi}$.
\end{proof}

We can generate an approximation cone for the pseudocodeword $\vomega$ by repeatedly transmitting a vector in $\mc K_\vomega$ (say, the sum of all constraint vectors of $\mc P(H)$ that are binding at $\vomega$) in a low signal-to-noise environment and recording the pseudocodewords that arise as the LP decoder's output.  If enough unique pseudocodewords are found, the insphere of the corresponding approximation cone ought to be a good approximation of the insphere of the recovery cone -- discussion regarding the corresponding block error rates can be found in Sections~\ref{sec:example} and~\ref{sec:performance}.  

In summary, Proposition~\ref{prop:coniccombo} states that any conic combination of constraint vectors that are active at the pseudocodeword $\vomega$ can be used to reliably transmit $\vomega$.  A lower bound for the probability of correctly recovering $\vomega$ over the AWGN channel can be maximized by transmitting a vector proportional to the center of the insphere of the recovery cone.  With this in mind, we define \emph{insphere modulation} as the method of mapping $\vomega$ to the center of an appropriate insphere, whether it is the insphere of the recovery cone or an approximation cone.

\section{A Method of Computing the Insphere of a Polyhedral Cone}\label{sec:insphere}

In this section we develop an iterative algorithm capable of approximating the insphere of a polyhedral cone to any degree of accuracy.  While our primary application will be to produce modulated versions of pseudocodewords that maximize a lower bound on the probability of block error, the results of this section may be of independent mathematical interest.  

Throughout this section, we use $B_r(\v x)$ to denote the closed sphere (or ball) of radius $r$ centered at $\v x$, i.e., $B_r(\v x) :=\{\v y \in \bbr^n \, | \, \norm{\v x - \v y} \leq r\}$.

\begin{defn}[\cite{Henrion2010}]\label{defn:insphere}
Let $\mc K$ be a polyhedral cone in $\bbr^n$.  The \emph{insphere} of $\mc K$ is the largest sphere contained in $\mc K$ whose center is in $B_1(\v 0)$, and the \emph{inradius} of $\mc K$ is the radius of the insphere.
\end{defn}

It is a consequence of Theorem 2.4 in~\cite{Henrion2010} that any polyhedral cone with a non-empty topological interior possesses a unique insphere.  Since polyhedral cones with empty interiors are of little use for transmitting pseudocodewords, we assume that our polyhedral cone is full dimensional and thus has a unique insphere.  Moreover, we observe that the center of such an insphere will necessarily have unit magnitude.

Throughout this section we will assume that we are in possession of an explicit description of $\mc K$ in terms of linear inequality constraints $K \v x \geq \v 0$. By scaling we can ensure that every row $\v k^T_i$ of $K$ has unit length - as such, $\v k^T_i \v x$ measures the Euclidean distance from the point $\v x$ to the hyperplane defined by the $i$th row of $K$. Following an idea similar to that employed to find inspheres of two-dimensional polyhedra given in~\cite{incircle}, the insphere can be found by maximizing a lower bound $z$ on the distance from a point in $\mc K \cap B_1(\v 0)$ to any defining hyperplane.  This idea is expressed succinctly by the following nonlinear convex optimization problem that, for reasons that will soon be become clear, we call Problem $P_\infty$.
\begin{align*}\label{prob:infty}\tag{$P_\infty$}
\text{max } & \, z \\
\text{subject to }& \, \v k^T_i \v x \geq z \text{ for all } i \\
 & \, \norm{\v x}^2 \leq 1
\end{align*}
If $(\v x_\infty, z_\infty)$ is an optimal solution to Problem $P_\infty$, then $\v x_\infty$ will give the center of the insphere and $z_\infty$ will give the inradius.  

The special structure of Problem $P_\infty$ admits an interesting algorithm that seeks to model the nonlinear boundary of the constraint $\norm{\v x}^2 = \sum_{i = 1}^n x_i^2 \leq 1$ by iteratively adding linear inequality constraints. Our algorithm operates by defining a sequence of linear programming problems and examining the optimal solutions of each.  At each stage, we also derive an upper and lower bounds on the inradius of $\mc K$.  The primary result of this section is that the sequence of optimal solutions converges to the center of the insphere.  Moreover, the upper and lower bounds on the inradius are easily computable and can be used to terminate computations once the desired degree of accuracy is attained.

To initialize our algorithm we form a linear program $P_0$ that contains all of the linear constraints of problem $P_\infty$ but eschews the nonlinear constraint in favor of $2n$ linear constraints that confine potential solutions to a bounded hypercube.  This is done to ensure that an optimal solution to this problem (and all subsequent problems) exist.  Formally, we defined problem $P_0$ as the following linear program:
\begin{align*}\label{prob:p0}\tag{$P_0$}
\text{max } & z \\
\text{subject to }& \v k^T_i \v x \geq z \text{ for all } i\\
 & -1 \leq x_i \leq 1 \text{ for all } i
\end{align*}

We formulate problem $P_{\ell+1}$ based off of problem $P_\ell$ as follows.  We first compute an optimal solution $(\v x_\ell, z_\ell)$ to problem $P_\ell$ (such a solution is readily obtained from any LP solver) and then compute the unit vector $\v u_\ell$ in the direction of $\v x_\ell$.  We note that the optimality of $\v x_\ell$ implies that $\norm{\v x_\ell} \geq 1$, which among other things implies that $\v x_\ell \not = \v 0$ and so $\v u_\ell$ is defined for all $\ell$.  Next, we record the value $w_\ell = \frac{z_\ell}{\norm{\v x_\ell}} = \min_i\{ \v k^T_i \v u_\ell  \}$, since $w_\ell$ is the radius of the largest sphere centered at the unit vector $\v u_\ell$ and contained within $\mc K$ and thus forms a lower bound on the inradius of $\mc K$.  We define Problem $P_{\ell + 1}$ by duplicating Problem $P_\ell$ and requiring that the additional constraint $ \v u_\ell^T \v x \leq 1$ also be met.  The general form of Problem $P_{\ell+1}$ is given as follows:
\begin{align*}\label{prob:plplus1}\tag{$P_{\ell+1}$}
\text{max } & z \\
\text{subject to }& \v k^T_i \v x \geq z \text{ for all } i\\
 & -1 \leq x_i \leq 1 \text{ for all } i \\
 & \v u_0^T \v x \leq 1 \\
 & \vdots \\
 & \v u_{\ell-1}^T \v x \leq 1 \\
 & \v u_\ell^T \v x \leq 1
\end{align*}
This process of iteratively approximating the insphere of a polyhedral cone is illustrated in Figure~\ref{fig:insphereIllustration}.

\begin{figure}
\centering
\begin{subfigure}[t]{.4\linewidth}
\includegraphics[width=\linewidth]{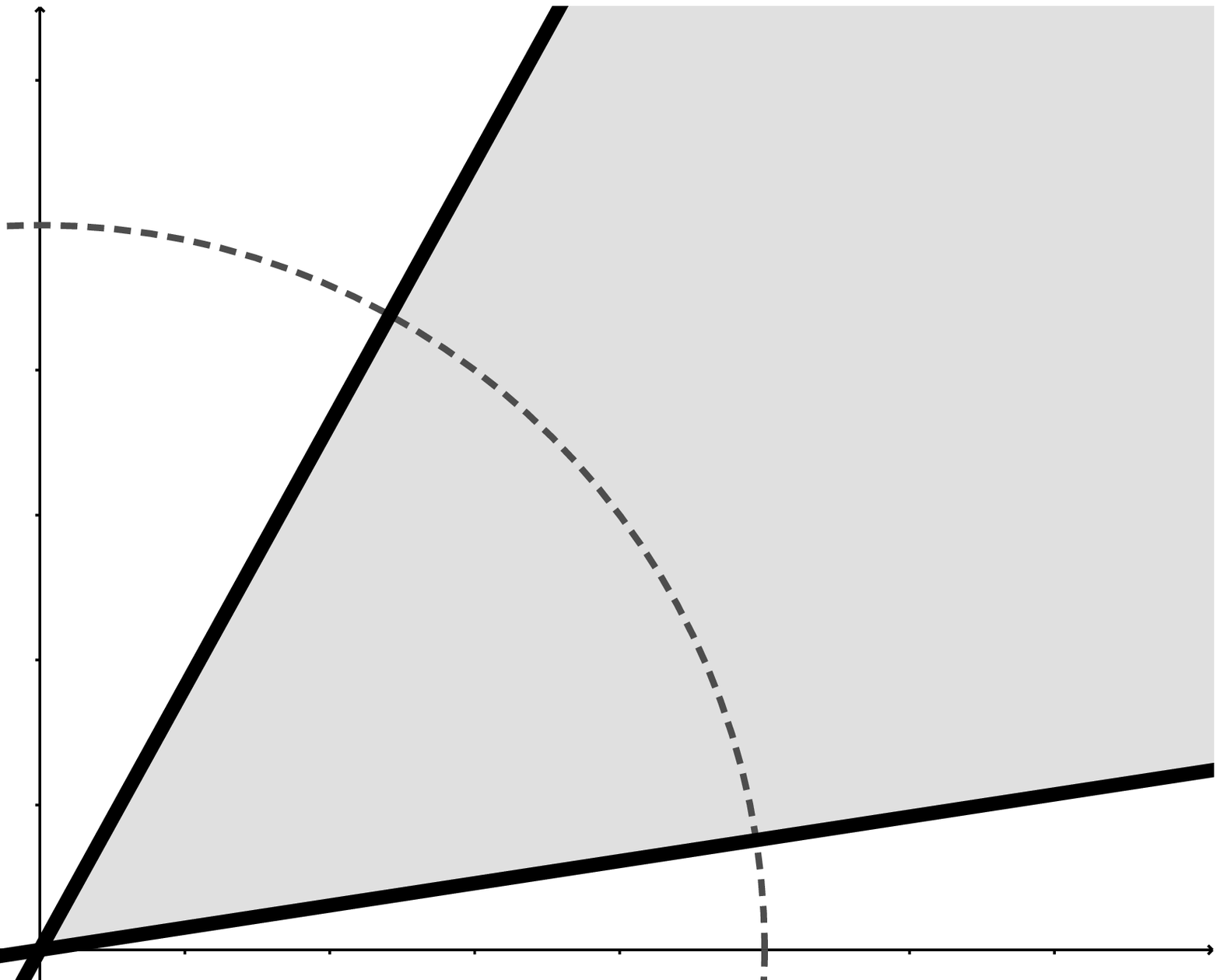}
\caption{A cone with a unit circle superimposed.}
\end{subfigure} \hspace{.5in}
\begin{subfigure}[t]{.4\linewidth}
\includegraphics[width=\linewidth]{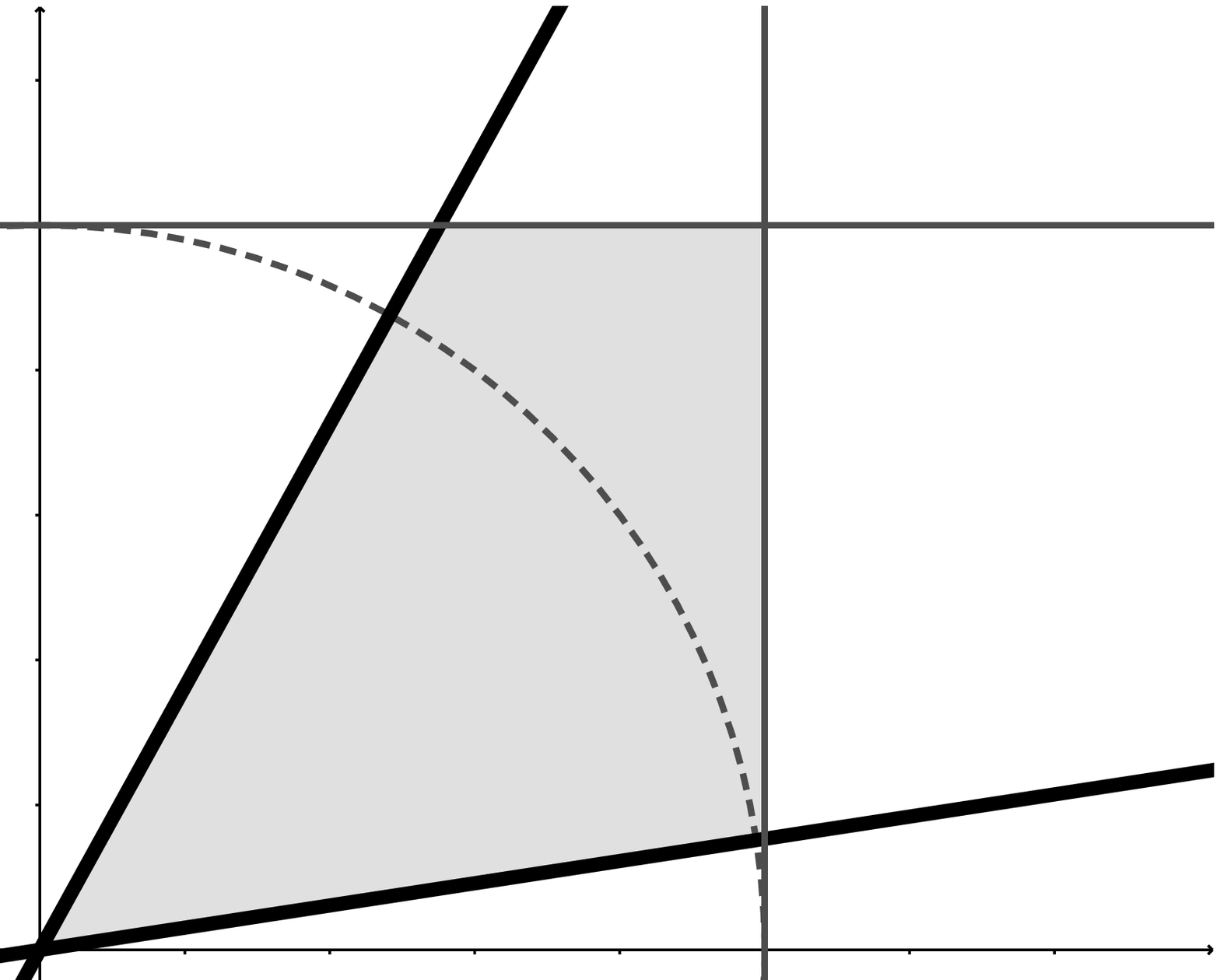}
\caption{Additional constraints are added to create the feasible region for Problem $P_0$.}
\end{subfigure}

\begin{subfigure}[t]{.4\linewidth}
\includegraphics[width=\linewidth]{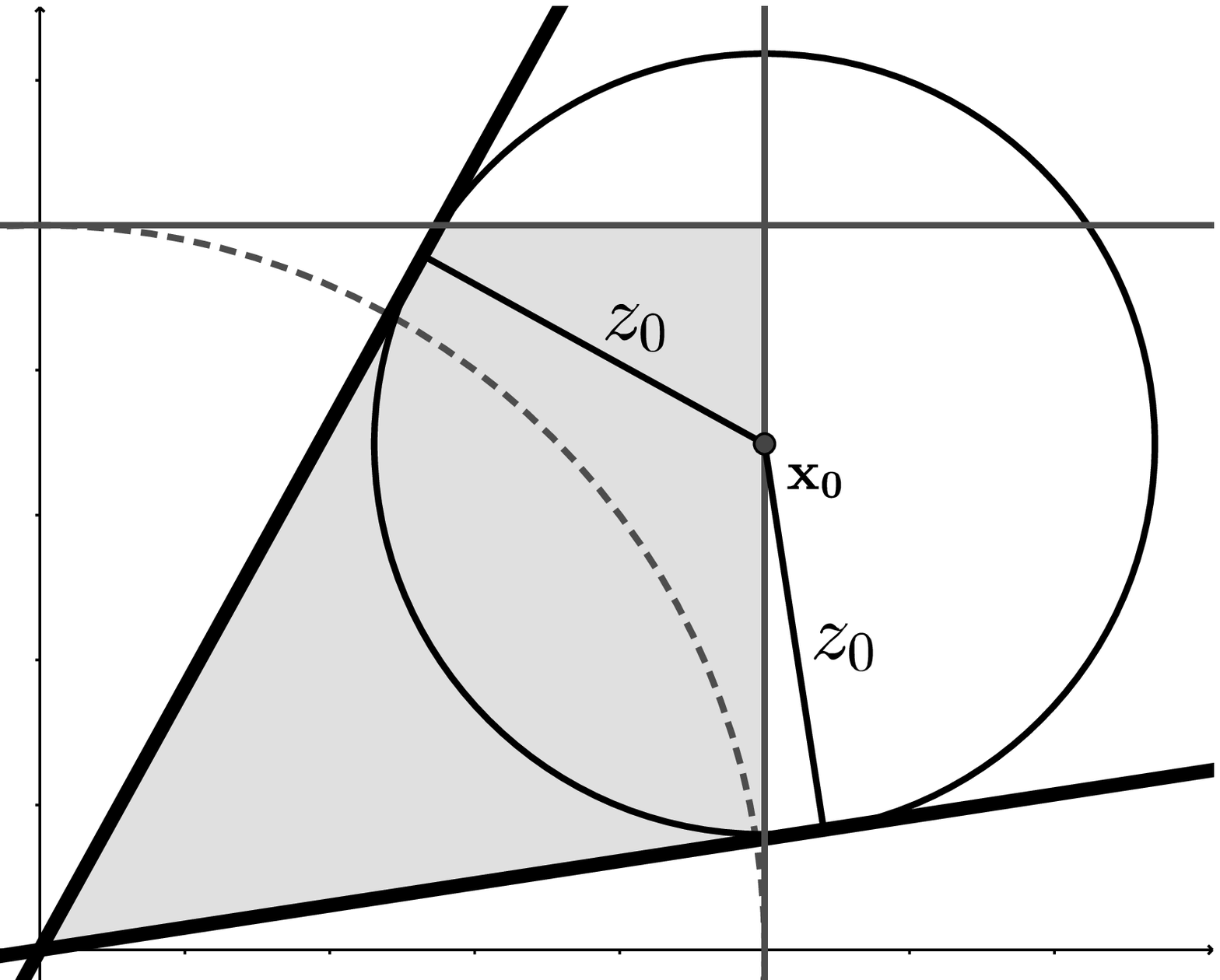}
\caption{The largest sphere contained in the cone and centered in Problem $P_0$'s feasible region is found.}
\end{subfigure}\hspace{.5in}
\begin{subfigure}[t]{.4\linewidth}
\includegraphics[width=\linewidth]{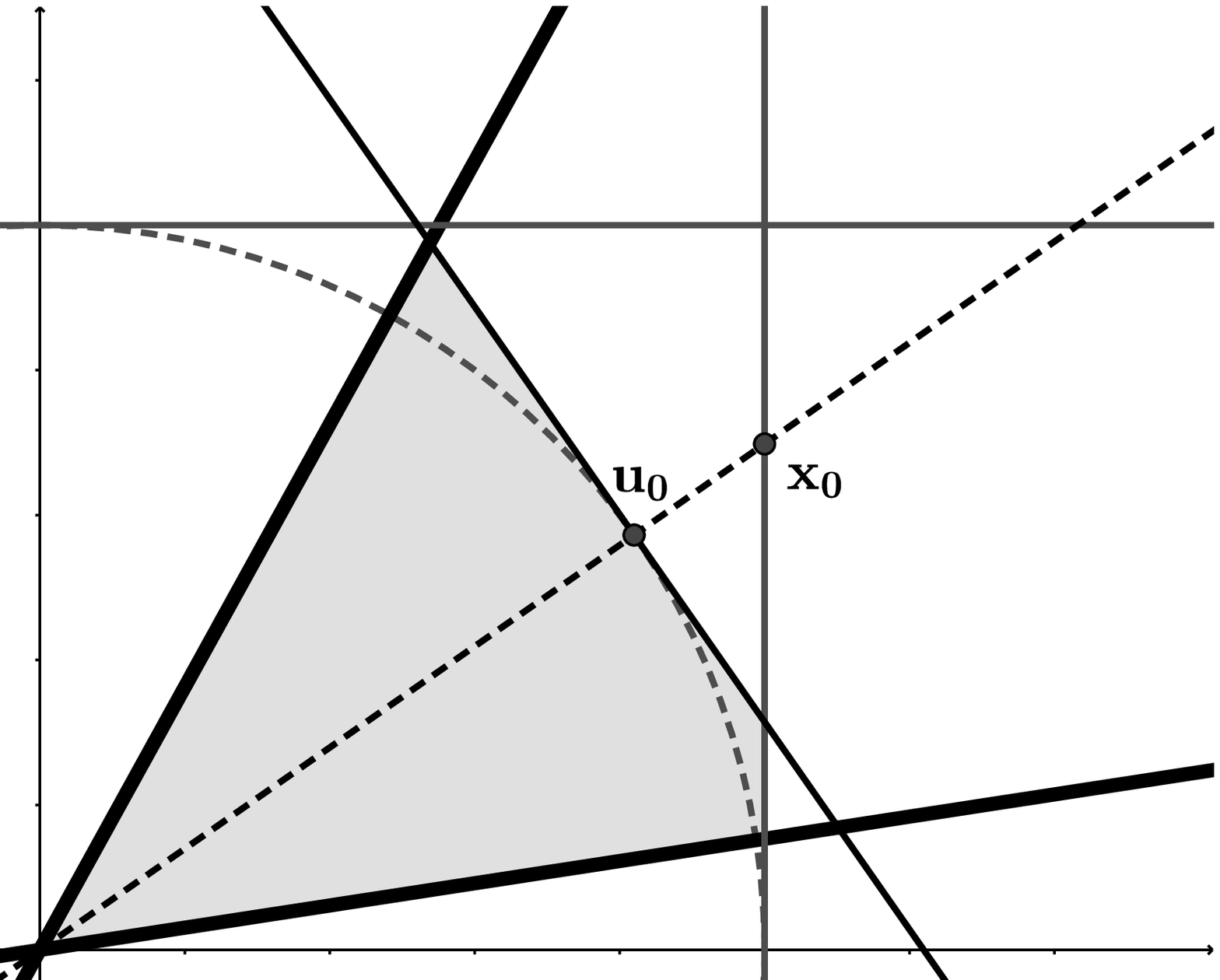}
\caption{A new constraint is added to form the feasible set of Problem $P_1$.}
\end{subfigure}

\begin{subfigure}[t]{.4\linewidth}
\includegraphics[width=\linewidth]{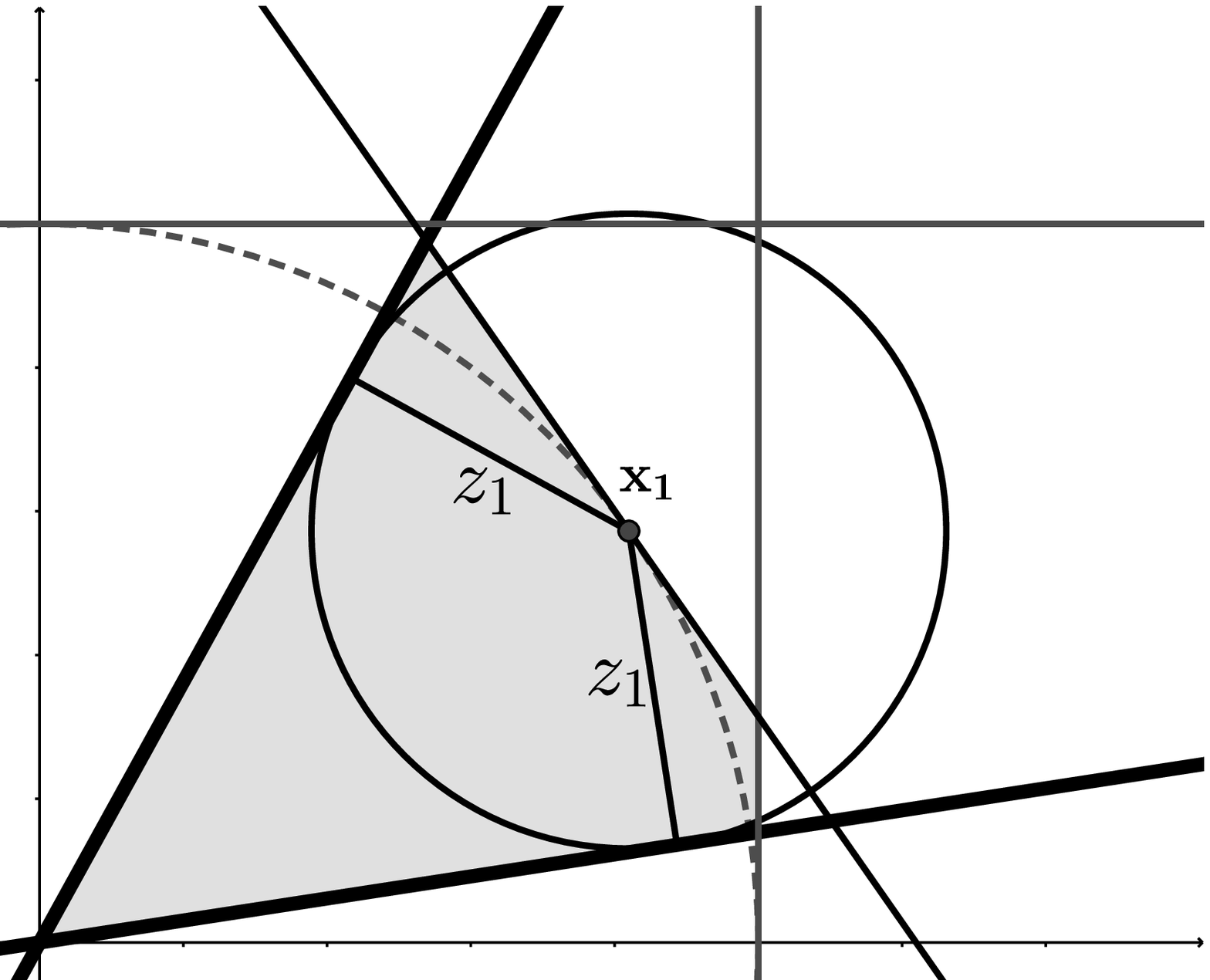}
\caption{The largest sphere contained in the cone and centered in Problem $P_1$'s feasible region is found.}
\end{subfigure}\hspace{.5in}
\begin{subfigure}[t]{.4\linewidth}
\includegraphics[width=\linewidth]{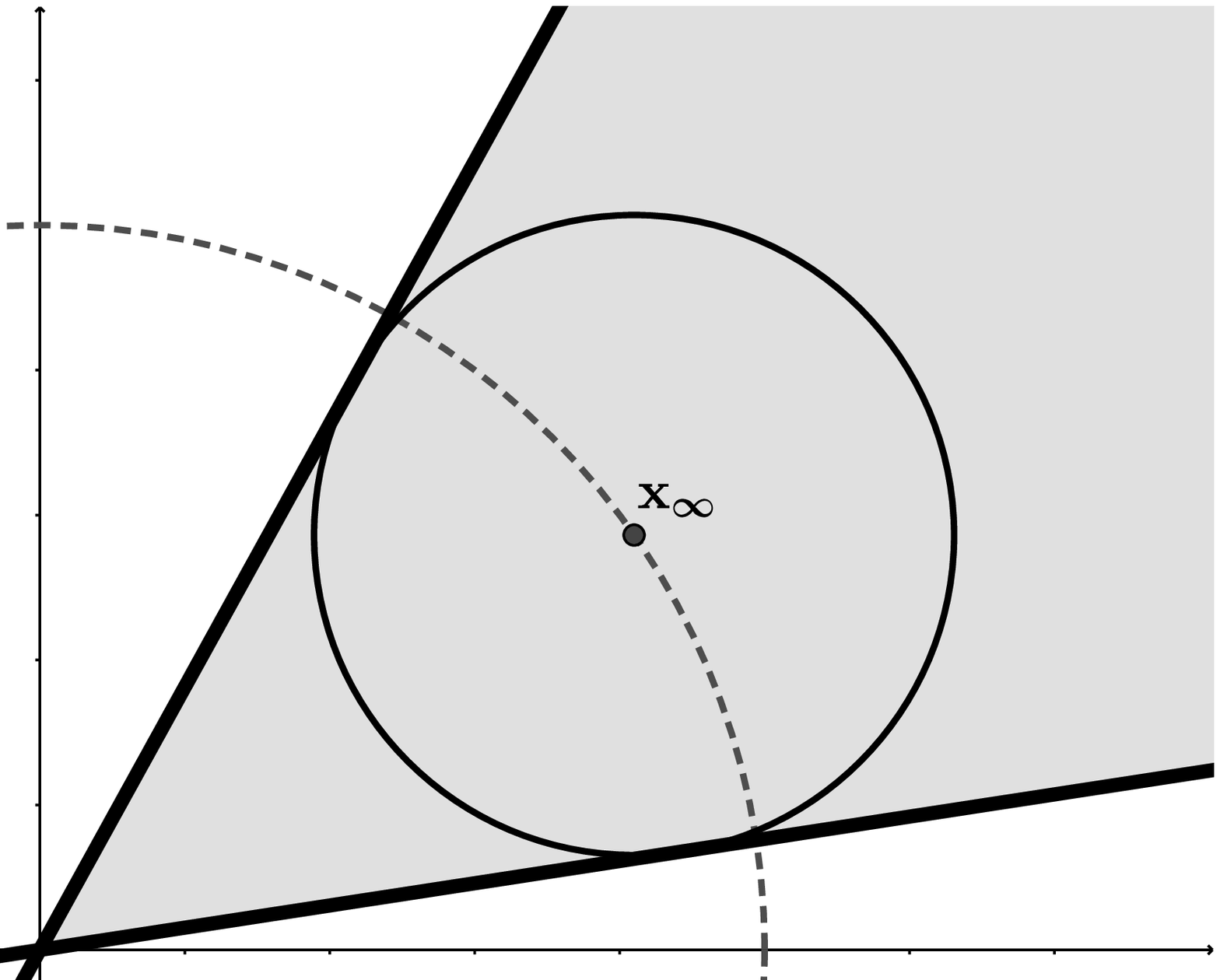}
\caption{Since $\v x_1$ lies on the unit circle, the insphere of the cone has been found.}
\end{subfigure}
\caption{An illustration of our iterative method for computing the insphere of a polyhedral cone to any degree of accuracy.  For this small example in $\bbr^2$, the exact insphere is found in only two steps.}
\label{fig:insphereIllustration}
\end{figure}

With these definitions one can see that the feasible sets for problems $P_0, P_1, \dots, P_\ell, \dots, P_\infty$ are nested within each other, forming a descending chain of sets with respect to inclusion.  Since all are maximization problems, it follows that the sequence of $z_\ell$'s is monotonically decreasing.  By comparison, the sequence of $w_\ell$'s has been observed to be generally increasing in nature, though not always monotonically (see Figure~\ref{fig:insphereIteration}).  Still, we know from a previous observation that $w_\ell$ forms a lower bound on $z_\infty$.  We obtain Proposition~\ref{prop:zzw} by combining these results.

\begin{proposition}\label{prop:zzw}
Let $\mc K$ be a polyhedral cone with non-empty topological interior and let $z_\infty$ denote its inradius. Let $(\v x_\ell, z_\ell)$ be an optimal solution to Problem $P_\ell$, and define $w_\ell = \frac{z_\ell}{\norm{\v x_\ell}}$.  With this notation, we have that
\[
z_0 \geq z_1 \geq \dots \geq z_\ell \geq \dots  \geq z_\infty.
\]  
Moreover, for any pair of integers $m$ and $n$ we have 
\[
z_m \geq z_\infty \geq w_n.
\]
\end{proposition}  

Proposition~\ref{prop:zzw} states that the inradius always lies between $w_\ell$ and $z_\ell$.  To compute the inradius to within $\delta$ of its true value, once can compute solutions to these problems until $z_\ell - w_\ell \leq \delta$.  Computational results, such as those summarized in Figure~\ref{fig:insphereIteration}, suggest that $z_\ell - w_\ell$ can always be driven to be arbitrarily small.  As we shall see from Theorem~\ref{thm:converge}, this is indeed the case. To prove this theorem, we first present a lemma.

\begin{figure}[h]
\centering
\includegraphics[scale = .7]{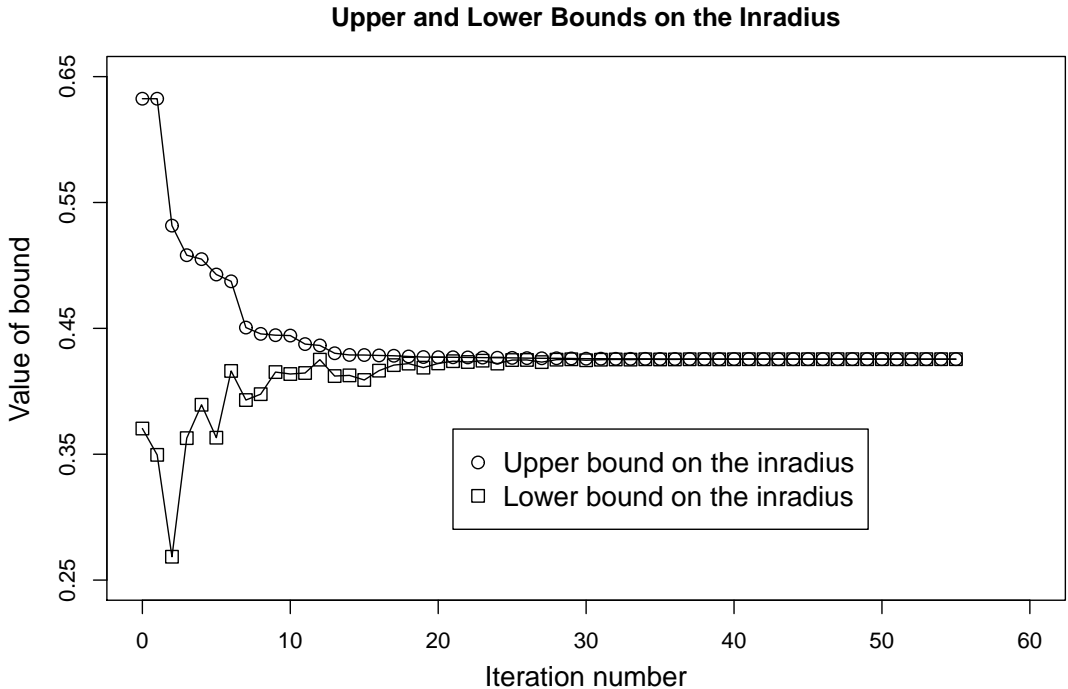}
\caption{A plot of upper bounds $z_\ell$ and lower bounds $w_\ell$ on the inradius for an approximation cone used to compute a transmittable version of the pseudocodeword $\vomega_3$ (see Section~\ref{sec:example}).  The cone is described by 35 constraints in $\bbr^{16}$.  In this example, it takes 56 iterations to ensure that the upper and lower bounds are within $10^{-5}$ of one another.}
\label{fig:insphereIteration}
\end{figure}

\begin{lemma}\label{lemma:ballCut}
Let $\mc K$ be a polyhedral cone with non-empty topological interior, let $(\v x_\ell, z_\ell)$ be an optimal solution to Problem $P_\ell$, and define $\v u_\ell = \frac{\v x_\ell}{\norm{\v x_\ell}}$.  For any $r$ with $0 \leq r < \norm{\v x_\ell - \v u_\ell}$, $B_r(\v x_\ell)$ and the feasible set for Problem $P_{\ell + 1}$ are disjoint.  
\end{lemma}

The proof of Lemma~\ref{lemma:ballCut} is virtually identical to the proof of the separating hyperplane theorem - see, e.g., pages 170-172 of~\cite{bt}.  This theorem states that given a non-empty, closed convex set $S$ in $\bbr^n$ and a point $\v x \not \in S$, there exists a hyperplane in $\bbr^n$ for which $\v x$ is on one side and the entirety of $S$ lies on the other.  Lemma~\ref{lemma:ballCut} merely makes this specific, stating that a small sphere around $\v x_\ell$ can be separated from $\v u_\ell$ by the inequality $\v u_\ell^T \v x \leq 1$, which is precisely the inequality added to Problem $P_{\ell}$ to form Problem $P_{\ell+1}$.  

\begin{theorem}\label{thm:converge}
Let $\mc K$ be a polyhedral cone with non-empty topological interior, and let  $\v x_\infty$ and $z_\infty$ denote the center of its insphere and its inradius, respectively. If $(\v x_\ell, z_\ell)$ is an optimal solution to Problem $P_\ell$, then $\v x_\ell$ converges to $\v x_\infty$ and $z_\ell$ converges to $z_\infty$.
\end{theorem}

\begin{proof}
Define $\v u_\ell := \frac{\v x_\ell}{\norm{\v x_\ell}}$ and suppose for the sake of a contradiction that $\v x_\ell - \v u_\ell$ does not converge to the zero vector.  This implies that there is an infinite subsequence $\ell_k$ and some $\epsilon_0 > 0$ so that $\norm{\v x_{\ell_k} - \v u_{\ell_k}} \geq \epsilon_0$ for all $k$.  Since the feasible sets of all Problems $P_\ell$ are contained within the feasible set for Problem $P_0$, which is bounded, the sequence $\v x_{\ell_k}$ is bounded.  The Bolzano-Weierstrass Theorem implies that there is a subsequence $\v x_{\ell_{k_j}}$ of $\v x_{\ell_k}$ that converges to some vector.  

Since  $\v x_{\ell_{k_j}}$ is a convergent sequence, it must be a Cauchy sequence.  Thus, there a $J$ so that for all $j', j'' \geq J$ we have $\norm{\v x_{\ell_{k_{j'}}} - \v x_{\ell_{k_{j''}}}} \leq \epsilon_0 / 4$.  The triangle inequality implies that for all $j' \geq J$ the vector $\v x_{\ell_{k_{j'}}}$ is in $B_{\epsilon_0 / 2}(\v x_{\ell_{k_{J}}})$.  Since $\norm{\v x_{\ell_{k_{j}}} - \v u_{\ell_{k_{j}}}} \geq \epsilon_0$ for all $j$, Lemma~\ref{lemma:ballCut} implies that $B_{\epsilon_0 / 2}(\v x_{\ell_{k_{J}}})$ is disjoint from the feasible set of Problem $P_{\ell_{j'}}$ for all $j' > J$, which implies that $\v x_{\ell_{k_{J + 1}}}$ is infeasible for Problem $P_{\ell_{k_{J+1}}}$.  This is a contradiction, since $\v x_{\ell_{k_{J + 1}}}$ was taken to be an optimal solution to Problem  $P_{\ell_{k_{J+1}}}$.  We conclude that $\v x_\ell - \v u_\ell$ converges to the zero vector.  

Since $\v x_\ell - \v u_\ell$ converges to the zero vector and the individual sequences $\v x_\ell$ and $\v u_\ell$ are both bounded, we can conclude that the individual sequences $\v x_\ell$ and $\v u_\ell$ converge to a common limit $\v x^\ast$, which is necessarily a unit vector.  With $w_\ell = \frac{z_\ell}{\norm{\v x_\ell}}$, the convergence of $\v x_\ell$ and $\v u_\ell$ to a common limit implies that $z_\ell$ and $w_\ell$ converge to a common limit, which by Proposition~\ref{prop:zzw} must be the inradius of $\mc K$.  This, along with the uniqueness of the insphere (see Theorem 2.4 in~\cite{Henrion2010}) implies that $\v x^\ast = \v x_\infty$. 
\end{proof}

%

\section{Extending The Modulation Scheme via $\mc C$-Symmetry}\label{sec:csym}

Using the results of Sections~\ref{sec:bestVector} and~\ref{sec:insphere}, one can obtain a transmitted vector that can reliably recover a linear programming pseudocodeword.  While finding such a vector takes some work, once computed it can be easily modified to allow for the broadcast of many more pseudocodewords.

In order to show that the probability of block error is independent of the codeword transmitted, the concept of \emph{$\mc C$-symmetry} is introduced in~\cite{Fel00}\footnote{Though the term ``$\mc C$-symmetry'' is never used in~\cite{fwk05}, the relevant concepts are present there as well.}.

\begin{defn}[\cite{Fel00}]\label{defn:relativePoint}
Let $\mc C$ be a code presented by the parity-check matrix $H$, and let $\v c \in \mc C$ be given.  For any point $\v x \in \mc P(H)$, the \emph{relative point} $\v x^{\v c}$ of $\v x$ with respect to $\v c$ is the point whose $i$th coordinate is given by $|x_i - c_i|$ for all $i = 1, 2, \dots, n$.
\end{defn}

\begin{theorem}[\cite{Fel00}]\label{thm:cSymmetry}
Let $\mc C$ be a code presented by the parity-check matrix $H$. For any point $\v x \in \mc P(H)$ and any codeword $\v c \in \mc C$, the relative point $\v x^{\v c}$ is also an element of $\mc P(H)$.  Further, if $\vomega$ is an extreme point of $\mc P(H)$, then $\vomega^{\v c}$ is also an extreme point of $\mc P(H)$.
\end{theorem}

Given a codeword $\v c \in \mc C$, define the map $\zeta_{\v c}: \mc P(H) \to \mc P(H)$  by $\zeta_{\v c}(\v x) = \v x^{\v c}$.  Since $\zeta_{\v c}$ is an isometry of $\bbr^n$, it is also an isometry of $\mc P(H)$.  Additionally, the map $\zeta_{\v c}$ defines a group action of the code $\mc C$ (viewed as an abelian group under vector addition) on $\mc P(H)$ as well as on the set of extreme points of $\mc P(H)$, i.e., the set of LP pseudocodewords.

For a row $\v h_j$ of $H$ and an odd-sized subset $S$ of $N(\v h_j)$, let $(j,S)$ denote the constraint of the fundamental polytope given by $\sum_{i \in S} x_i + \sum_{i' \in N(\v h_j)\setminus S}(1- x_{i'}) \leq |N(\v h_j)| - 1$, which we rewrite as $-\sum_{i \in S} x_i + \sum_{i' \in N(\v h_j)\setminus S} x_{i'} \geq 1 - |S|$. For a codeword $\v c \in \mc C$, let $S^{\v c} := S \triangle (\supp(\v c) \cap N(\v h_j))$, where the ``$\triangle$'' operator denotes the symmetric difference of two sets. By following these definitions, one can prove the following proposition:

\begin{proposition}\label{prop:activeOrbit}
Let $\mc C$ be a binary code presented by the parity-check matrix $H$ with fundamental polytope $\mc P(H)$. Let $j$ be the index of a row $\v h_j$ of $H$ and $S$ be an odd-sized subset of $N(\v h_j)$, and let $\v x \in \mc P(H)$ and $\v c \in C$ be given.  If the constraint defined by the pair $(j,S)$ is active at $\v x$, then the constraint $(j,S^{\v c})$ is active at the relative point $\v x^{\v c}$.
\end{proposition}

Using Proposition~\ref{prop:activeOrbit}, we see that the set of constraint vectors of the fundamental polytope that are active at $\vomega^{\v c}$ can be obtained by first taking the set of constraints vectors that are active at $\vomega$ and multiplying each component in the support of $\v c$ by -1.  Moreover, if $\Psi = \{\vpsi_1, \vpsi_2, \dots, \vpsi_r\}$ is a set of LP pseudocodewords and $\Psi^{\v c} := \{\vpsi_1^{\v c}, \vpsi_2^{\v c}, \dots, \vpsi_r^{\v c}\}$ , then the approximation cone $\mc R_{\vomega^{\v c}, \Psi^{\v c}}$ can be obtained by applying the same transformation to $\mc R_{\vomega, \Psi}$. These facts are summarized in the following proposition. 

\begin{proposition}\label{prop:tSymmetry}
Let $\mc C$ be a code presented by the parity-check matrix $H$, let $\v c$ be a codeword and $\vomega$ be a linear programming pseudocodeword, and let $\Psi = \{\vpsi_1, \vpsi_2, \dots, \vpsi_r\}$ be a set of linear programming pseudocodewords. The map $\phi_{\v c}: \bbr^n\to \bbr^n$ defined by $x_i \to (-1)^{c_i} x_i$ for all $i = 1, 2, \dots, n$ is an isometry between the recovery cones $\mc K_\vomega$ and $\mc K_{\vomega^{\v c}}$ as well as an isometry between the approximation cones $\mc R_{\vomega, \Psi}$ and $\mc R_{\vomega^{\v c}, \Psi^{\v c}}$.
\end{proposition}

Since the probability density function for additive white Gaussian noise with a given power spectral density is symmetric about any such isometry defined in Proposition~\ref{prop:tSymmetry}, we have the following result.

\begin{theorem}\label{thm:transmit}
Let $\mc C$ be a code presented by the parity-check matrix $H$, let $\v c$ be a codeword and $\vomega$ be a linear programming pseudocodeword, and define the map $\phi_{\v c}: \bbr^n\to \bbr^n$ by $x_i \to (-1)^{c_i} x_i$ for all $i = 1, 2, \dots, n$.  If $t(\vomega)$ is a vector in the recovery cone $\mc K_\vomega$ suitable for transmitting $\vomega$ over the additive white Gaussian noise channel, then $\phi_{\v c}\bigg(t(\vomega)\bigg)$ is a vector in $\mc K_{\vomega^{\v c}}$ suitable for transmitting $\vomega^{\v c}$. Moreover, the probability of correctly decoding to $\vomega$ when $t(\vomega)$ is transmitted is equal to the probability of correctly decoding to $\vomega^{\v c}$ when $\phi_{\v c}\bigg(t(\vomega)\bigg)$ is transmitted.
\end{theorem}
Theorem~\ref{thm:transmit} has a practical consequence:  once we have found a vector that provides a low probability of error in recovering the pseudocodeword $\vomega$, we can simply apply the map $\phi_{\v c}$ to this vector to transmit another pseudocodeword in $\vomega$'s $\mc C$-symmetric orbit.

\section{A Detailed Example}\label{sec:example}

In this section we provide a detailed example of a code, the orbits of its pseudocodewords, the modulated vectors used to transmit these pseudocodewords, and the probability of block error when using said vectors.  We use a \emph{cycle code} because a complete characterization of their linear programming pseudocodewords is given in~\cite{axvigDreher}, which allows for easy access to a pool of known pseudocodewords.

\begin{defn}\label{defn:cycleCode}
A \emph{cycle code} is a binary linear code $\mc C$ equipped with a parity-check matrix $H$ with a uniform column weight of 2.
\end{defn}

Given a parity-check matrix $H$ of a cycle code, one can view $H$ as the vertex-edge incidence matrix of a graph, the so-called \emph{normal graph} $N$ of the code.  It is an easy exercise to show that a binary vector $\v c$ is in the null space of $H$ if and only if the support of $\v c$ indexes an edge-disjoint union of cycles in $N$.  For our example, we will consider the cycle code $\mc C_1$ of length 16 and dimension 5 presented as the null space of the parity-check matrix
\[
\setcounter{MaxMatrixCols}{20}
H_1 = \footnotesize{\begin{bmatrix}
1 & 0 & 0 & 1 & 0 & 0 & 0 & 0 & 0 & 0 & 0 & 0 & 0 & 0 & 0 & 0 \\
1 & 1 & 0 & 0 & 1 & 0 & 0 & 0 & 0 & 0 & 0 & 0 & 0 & 0 & 0 & 0 \\
0 & 1 & 1 & 0 & 0 & 1 & 0 & 0 & 0 & 0 & 0 & 0 & 0 & 0 & 0 & 0 \\
0 & 0 & 1 & 0 & 0 & 0 & 1 & 0 & 0 & 0 & 0 & 0 & 0 & 0 & 0 & 0 \\
0 & 0 & 0 & 1 & 1 & 0 & 0 & 1 & 0 & 0 & 0 & 0 & 0 & 0 & 0 & 0 \\
0 & 0 & 0 & 0 & 0 & 1 & 1 & 0 & 1 & 0 & 0 & 0 & 0 & 0 & 0 & 0 \\
0 & 0 & 0 & 0 & 0 & 0 & 0 & 1 & 0 & 1 & 1 & 0 & 0 & 0 & 0 & 0 \\
0 & 0 & 0 & 0 & 0 & 0 & 0 & 0 & 1 & 0 & 0 & 1 & 1 & 0 & 0 & 0 \\
0 & 0 & 0 & 0 & 0 & 0 & 0 & 0 & 0 & 1 & 0 & 0 & 0 & 1 & 0 & 0 \\
0 & 0 & 0 & 0 & 0 & 0 & 0 & 0 & 0 & 0 & 1 & 0 & 0 & 1 & 1 & 0 \\
0 & 0 & 0 & 0 & 0 & 0 & 0 & 0 & 0 & 0 & 0 & 1 & 0 & 0 & 1 & 1 \\
0 & 0 & 0 & 0 & 0 & 0 & 0 & 0 & 0 & 0 & 0 & 0 & 1 & 0 & 0 & 1 \\
\end{bmatrix}}.
\]
When viewed as the vertex-edge incidence matrix for a graph $N_1$ with vertices $u_1, u_2, \dots, u_{12}$ and edges $x_1, x_2, \dots, x_{16}$, we obtain the normal graph $N_1$ in Figure~\ref{fig:normalGraph}.
\begin{figure}
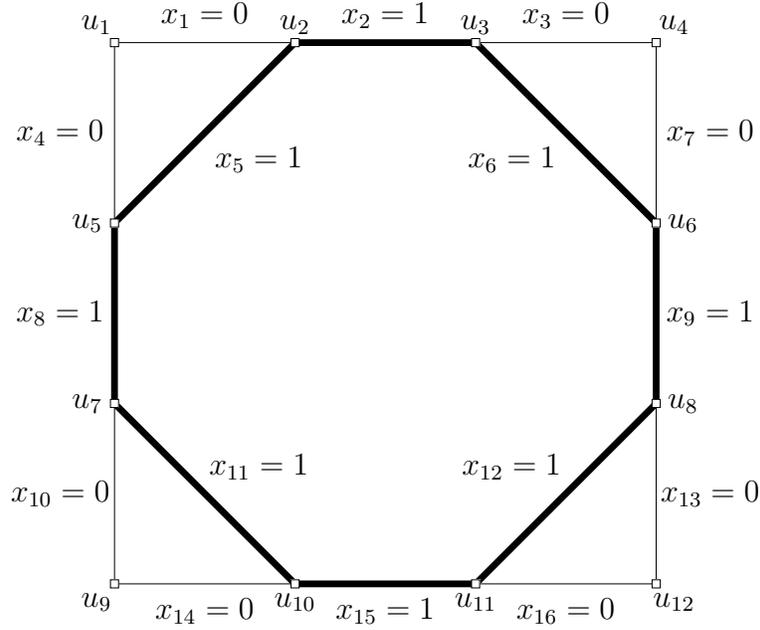

\centering
$\begin{graph}(8.5,8)
\unitlength=.6cm

\squarenode{f1}(1,13)[\graphnodecolour(1,1,1)]
\squarenode{f2}(5,13)[\graphnodecolour(1,1,1)]
\squarenode{f3}(9,13)[\graphnodecolour(1,1,1)]
\squarenode{f4}(13,13)[\graphnodecolour(1,1,1)]
\squarenode{f5}(1,9)[\graphnodecolour(1,1,1)]
\squarenode{f6}(13,9)[\graphnodecolour(1,1,1)]
\squarenode{f7}(1,5)[\graphnodecolour(1,1,1)]
\squarenode{f8}(13,5)[\graphnodecolour(1,1,1)]
\squarenode{f9}(1,1)[\graphnodecolour(1,1,1)]
\squarenode{f10}(5,1)[\graphnodecolour(1,1,1)]
\squarenode{f11}(9,1)[\graphnodecolour(1,1,1)]
\squarenode{f12}(13,1)[\graphnodecolour(1,1,1)]


\roundnode{x1}(3,13)[\graphnodesize{0}]
\roundnode{x2}(7,13)[\graphnodesize{0}]
\roundnode{x3}(11,13)[\graphnodesize{0}]
\roundnode{x4}(1,11)[\graphnodesize{0}]
\roundnode{x5}(3,11)[\graphnodesize{0}]
\roundnode{x6}(11,11)[\graphnodesize{0}]
\roundnode{x7}(13,11)[\graphnodesize{0}]
\roundnode{x8}(1,7)[\graphnodesize{0}]
\roundnode{x9}(13,7)[\graphnodesize{0}]
\roundnode{x10}(1,3)[\graphnodesize{0}]
\roundnode{x11}(3,3)[\graphnodesize{0}]
\roundnode{x12}(11,3)[\graphnodesize{0}]
\roundnode{x13}(13,3)[\graphnodesize{0}]
\roundnode{x14}(3,1)[\graphnodesize{0}]
\roundnode{x15}(7,1)[\graphnodesize{0}]
\roundnode{x16}(11,1)[\graphnodesize{0}]

\edge{f1}{f2}
\edge{f1}{f5}
\edge{f2}{f3}[\graphlinewidth{0.15}]
\edge{f2}{f5}[\graphlinewidth{0.15}]
\edge{f3}{f4}
\edge{f3}{f6}[\graphlinewidth{0.15}]
\edge{f4}{f6}
\edge{f5}{f7}[\graphlinewidth{0.15}]
\edge{f6}{f8}[\graphlinewidth{0.15}]
\edge{f7}{f9}
\edge{f7}{f10}[\graphlinewidth{0.15}]
\edge{f8}{f11}[\graphlinewidth{0.15}]
\edge{f8}{f12}
\edge{f9}{f10}
\edge{f10}{f11}[\graphlinewidth{0.15}]
\edge{f11}{f12}

\nodetext{x1}(0,.6){$x_1 = 0$}
\nodetext{x2}(0,.6){$x_2 = 1$}
\nodetext{x3}(0,.6){$x_3 = 0$}
\nodetext{x4}(-1.2,0){$x_4 = 0$}
\nodetext{x5}(1.2,-.6){$x_5 = 1$}
\nodetext{x6}(-1.2,-.6){$x_6 = 1$}
\nodetext{x7}(1.2,0){$x_7 = 0$}
\nodetext{x8}(-1.2,0){$x_8 = 1$}
\nodetext{x9}(1.2,0){$x_9 = 1$}
\nodetext{x10}(-1.2,0){$x_{10} = 0$}
\nodetext{x11}(1.2,.6){$x_{11} = 1$}
\nodetext{x12}(-1.2,.6){$x_{12} = 1$}
\nodetext{x13}(1.2,0){$x_{13} = 0$}
\nodetext{x14}(0,-.6){$x_{14} = 0$}
\nodetext{x15}(0,-.6){$x_{15} = 1$}
\nodetext{x16}(0,-.6){$x_{16} = 0$}

\nodetext{f1}(-.4,.4){$u_1$}
\nodetext{f2}(0,.4){$u_2$}
\nodetext{f3}(0,.4){$u_3$}
\nodetext{f4}(.4,.4){$u_4$}
\nodetext{f5}(-.6,0){$u_5$}
\nodetext{f6}(.6,0){$u_6$}
\nodetext{f7}(-.6,0){$u_7$}
\nodetext{f8}(.6,0){$u_8$}
\nodetext{f9}(-.4,-.4){$u_9$}
\nodetext{f10}(0,-.4){$u_{10}$}
\nodetext{f11}(0,-.4){$u_{11}$}
\nodetext{f12}(.4,-.4){$u_{12}$}

\end{graph}
$
\caption{The normal graph $N_1$ corresponding to the parity-check matrix $H_1$.  The edges corresponding to the support of the codeword $\v c = [0, 1, 0, 0, 1, 1, 0, 1, 1, 0, 1, 1, 0, 0, 1, 0]^T$ are given in bold to emphasize the correspondence between edge-disjoint unions of cycles in the normal graph and codewords.}
\label{fig:normalGraph}
\end{figure}

As mentioned before, cycle codes derive their name from the bijective correspondence between binary codewords and edge-disjoint unions of cycles in the normal graph.  In~\cite{axvigDreher}, a graph-based characterization of both trivial and nontrivial LP pseudocodewords for cycle codes is given.

\begin{theorem}[\cite{axvigDreher}]\label{thm:extremepointschar}
Let $C$ be a cycle code of length $n$ presented by the parity-check matrix $H$, and let $\mc P$ and $N$ denote its fundamental polytope and normal graph, respectively. A vector $\v x \in \mc P$ is a linear programming pseudocodeword of $\mc P$ if and only if the following three conditions hold:
\begin{enumerate}[(a)]
\item $\v x \in \{0,\frac 12, 1\}^{n}$.

\item  With $\mc H_{\v x}$ defined to be $\{ i \, | \, x_i = \frac 12 \}$, the subgraph $\Gamma$ of $N$ induced by the edges corresponding to the indices in $\mc H_{\v x}$ is 2-regular.  Equivalently, $\Gamma$ is a union of vertex-disjoint simple cycles $\gamma_1, \gamma_2, \dots, \gamma_\ell$.

\item With $U^{\mathrm{o}}_{\gamma_i}$ defined as set of vertices in $\gamma_i$ that are incident with an odd number of edges assigned a value of 1 by $\v x$, the cardinality of $U^{\mathrm{o}}_{\gamma_i}$ is itself odd for all simple cycles $\gamma_i$ in $\Gamma$.

\end{enumerate}
\end{theorem}

Using Theorem~\ref{thm:extremepointschar}, one can see that the vector $\vomega_1$ given by
\[
\vomega_1 = \begin{bmatrix}\frac 12 & 1 & 1 & \frac 12 & \frac 12 & 0 & 1 & 0 & 1 & 0 & 0 & \frac 12 & \frac 12 & 0 & 0 & \frac 12 \end{bmatrix}^T
\]
must be a linear programming pseudocodeword - Figure~\ref{fig:pcw1} gives a graphical representation.  We note that, loosely speaking, $\vomega_1$ consists of a pair of cycles along with a path linking these cycles.  The edges on the cycles are all labeled with $\frac{1}{2}$, and those on the path between the cycles carry a value of 1.  A similar structure holds for all other pseudocodewords of cycle codes, and this structure is exploited by the \emph{cycle-path method} of~\cite{axvigDreher}, allowing for the explicit construction of LP pseudocodewords for cycle codes.

\begin{figure}
\centering
$\begin{graph}(8.5,8)
\unitlength=.6cm

\squarenode{f1}(1,13)[\graphnodecolour(1,1,1)]
\squarenode{f2}(5,13)[\graphnodecolour(1,1,1)]
\squarenode{f3}(9,13)[\graphnodecolour(1,1,1)]
\squarenode{f4}(13,13)[\graphnodecolour(1,1,1)]
\squarenode{f5}(1,9)[\graphnodecolour(1,1,1)]
\squarenode{f6}(13,9)[\graphnodecolour(1,1,1)]
\squarenode{f7}(1,5)[\graphnodecolour(1,1,1)]
\squarenode{f8}(13,5)[\graphnodecolour(1,1,1)]
\squarenode{f9}(1,1)[\graphnodecolour(1,1,1)]
\squarenode{f10}(5,1)[\graphnodecolour(1,1,1)]
\squarenode{f11}(9,1)[\graphnodecolour(1,1,1)]
\squarenode{f12}(13,1)[\graphnodecolour(1,1,1)]


\roundnode{x1}(3,13)[\graphnodesize{0}]
\roundnode{x2}(7,13)[\graphnodesize{0}]
\roundnode{x3}(11,13)[\graphnodesize{0}]
\roundnode{x4}(1,11)[\graphnodesize{0}]
\roundnode{x5}(3,11)[\graphnodesize{0}]
\roundnode{x6}(11,11)[\graphnodesize{0}]
\roundnode{x7}(13,11)[\graphnodesize{0}]
\roundnode{x8}(1,7)[\graphnodesize{0}]
\roundnode{x9}(13,7)[\graphnodesize{0}]
\roundnode{x10}(1,3)[\graphnodesize{0}]
\roundnode{x11}(3,3)[\graphnodesize{0}]
\roundnode{x12}(11,3)[\graphnodesize{0}]
\roundnode{x13}(13,3)[\graphnodesize{0}]
\roundnode{x14}(3,1)[\graphnodesize{0}]
\roundnode{x15}(7,1)[\graphnodesize{0}]
\roundnode{x16}(11,1)[\graphnodesize{0}]

\edge{f1}{f2}[\graphlinewidth{0.15}\graphlinedash{2.5}]
\edge{f1}{f5}[\graphlinewidth{0.15}\graphlinedash{2.5}]
\edge{f2}{f3}[\graphlinewidth{0.15}]
\edge{f2}{f5}[\graphlinewidth{0.15}\graphlinedash{2.5}]
\edge{f3}{f4}[\graphlinewidth{0.15}]
\edge{f3}{f6}
\edge{f4}{f6}[\graphlinewidth{0.15}]
\edge{f5}{f7}
\edge{f6}{f8}[\graphlinewidth{0.15}]
\edge{f7}{f9}
\edge{f7}{f10}
\edge{f8}{f11}[\graphlinewidth{0.15}\graphlinedash{2.5}]
\edge{f8}{f12}[\graphlinewidth{0.15}\graphlinedash{2.5}]
\edge{f9}{f10}
\edge{f10}{f11}
\edge{f11}{f12}[\graphlinewidth{0.15}\graphlinedash{2.5}]

\nodetext{x1}(0,.6){$x_1 = \frac 12$}
\nodetext{x2}(0,.6){$x_2 = 1$}
\nodetext{x3}(0,.6){$x_3 = 1$}
\nodetext{x4}(-1.2,0){$x_4 = \frac 12$}
\nodetext{x5}(1.2,-.6){$x_5 = \frac 12$}
\nodetext{x6}(-1.2,-.6){$x_6 = 0$}
\nodetext{x7}(1.2,0){$x_7 = 1$}
\nodetext{x8}(-1.2,0){$x_8 = 0$}
\nodetext{x9}(1.2,0){$x_9 = 1$}
\nodetext{x10}(-1.2,0){$x_{10} = 0$}
\nodetext{x11}(1.2,.6){$x_{11} = 0$}
\nodetext{x12}(-1.2,.6){$x_{12} = \frac 12$}
\nodetext{x13}(1.2,0){$x_{13} = \frac 12$}
\nodetext{x14}(0,-.6){$x_{14} = 0$}
\nodetext{x15}(0,-.6){$x_{15} = 0$}
\nodetext{x16}(0,-.6){$x_{16} = \frac 12$}

\nodetext{f1}(-.4,.4){$u_1$}
\nodetext{f2}(0,.4){$u_2$}
\nodetext{f3}(0,.4){$u_3$}
\nodetext{f4}(.4,.4){$u_4$}
\nodetext{f5}(-.6,0){$u_5$}
\nodetext{f6}(.6,0){$u_6$}
\nodetext{f7}(-.6,0){$u_7$}
\nodetext{f8}(.6,0){$u_8$}
\nodetext{f9}(-.4,-.4){$u_9$}
\nodetext{f10}(0,-.4){$u_{10}$}
\nodetext{f11}(0,-.4){$u_{11}$}
\nodetext{f12}(.4,-.4){$u_{12}$}

\end{graph}
$
\caption{A graphical representation of the pseudocodeword $\vomega_1$.}
\label{fig:pcw1}
\end{figure}

We now turn our attention to the task of modulating pseudocodewords; that is, mapping them to vectors suitable for transmission.  Unless stated otherwise the terms ``modulated vector'' and ``modulated pseudocodeword'' are taken to mean the center of an insphere for an appropriately chosen approximation cone.

The modulated vector $t(\vomega_1)$ suitable for transmitting $\vomega_1$ across the additive white Gaussian noise channel was found by a randomized, iterative procedure.  The set of pseudocodewords $\Psi_1$ was initialized as the empty set and the vector $\v f_1$ was set to be a conic combination of constraint vectors of $\mc P$ that were active at $\vomega_1$.  From here, each iteration consisted of the following steps.  Gaussian noise was repeatedly added to $\v f_{i}$ and the resulting vectors were fed into the LP decoder.  The pseudocodewords obtained as output vectors were recorded, and the unique polyhedral neighbors\footnote{The pseudocodeword $\vpsi$ is a polyhedral neighbor of $\vomega_1$ if both $\vpsi$ and $\vomega_1$ comprise the two end points of one of the fundamental polytope's edges.  This occurs precisely when the rank of constraint vectors that are active at \emph{both} vectors is one less than the dimension of the ambient space.  We restricted our sets $\Psi_i$ to such neighboring pseudocodewords since they are the most likely erroneous outputs of the LP decoder in low-noise scenarios.} of $\vomega_1$ were added to the set $\Psi_i$. Following the algorithm outlined in Section~\ref{sec:insphere}, the inradius of $\mc R_{\vomega_1, \Psi_i}$ was computed to within $10^{-5}$.  The conic combination $\v f_{i+1}$ of constraints used for the next iteration was then updated to be the center of the insphere of $\mc R_{\vomega_1, \Psi_i}$, and the set of pseudocodewords $\Psi_{i+1}$ was initialized at $\Psi_i$.  

Below is the center of the insphere for $\mc R_{\vomega_1, \Psi_{10}}$, recorded to two decimal places for ease of viewing:
\[
t(\vomega_1) =  \footnotesize{\begin{bmatrix}0 & -0.35&  -0.25&    0  &   0&  0.35&  -0.26&  0.35&  -0.35&  0.26&  0.35   &   0   &   0 & 0.26 & 0.35  &    0\end{bmatrix}}^T.
\]
The corresponding inradius was 0.49; however, we note that since this is the inradius of an approximation cone and not that of the true recovery cone it only serves to provide an upper bound on the true inradius.  We may therefore use such inradii only as rough gauges of error-correcting performance. For a concrete measure we rely on simulations (see Section~\ref{sec:performance}).

Following Section~\ref{sec:csym}, we can now apply $\mc C$-symmetry to $\vomega_1$ to obtain additional pseudocodewords as well as to $t(\vomega_1)$ to obtain modulated versions of these pseudocodewords.  Using $\mc C$-symmetry about the codeword $\v c$ of Figure~\ref{fig:normalGraph} we obtain
\[
\vomega_1^{\v c}  = \begin{bmatrix}\frac 12 & 0 & 1 & \frac 12 & \frac 12 & 1 & 1 & 1 & 0 & 0 & 1 & \frac 12 & \frac 12 & 0 & 1 & \frac 12 \end{bmatrix}^T
\]
and
\[
t(\vomega_1^{\v c})  = \footnotesize{\begin{bmatrix}0 & 0.35&  -0.25&    0  &   0&  -0.35&  -0.26&  -0.35&  0.35&  0.26&  -0.35   &   0   &   0 & 0.26 & -0.35  &    0\end{bmatrix}}^T.
\]
There are more pseudocodewords in the $\mc C$-symmetry orbit of $\vomega_1$; however, there are only $2^3 = 8$ in total.  This can be seen by observing that the operation of $\mc C$ symmetry will never affect any coordinate of a pseudocodeword which is assigned a value of $\frac 12$.  Since the locations where $\vomega_1$ is equal to $\frac 12$ consists of two vertex disjoint cycles in the normal graph we can conclude that $\vomega_1$'s stabilizer subgroup has size $2^2 = 4$ and thus its orbit has size $2^5/2^2 = 8$.  

Applying Theorem~\ref{thm:extremepointschar} and the cycle-path method of~\cite{axvigDreher}, we see that the linear programming pseudocodewords of this cycle code can be partitioned into 8 disjoint orbits under $\mc C$ symmetry:  7 orbits of nontrivial pseudocodewords and a separate orbit for the binary codewords.  A set of representatives for these orbits is given as follows:
\begin{align*}
\vomega_1 & = \begin{bmatrix} \frac{1}{2} &    1 &     1 &   \frac{1}{2} &  \frac{1}{2} &    0 &   1 &     0 &   1 &      0 &    0 &  \frac{1}{2} &   \frac{1}{2} &     0 &    0 &  \frac{1}{2}  \end{bmatrix}^T \\
\vomega_2 & = \begin{bmatrix} 1 &     1 &   \frac{1}{2} &    1 &     0 & \frac{1}{2} &  \frac{1}{2} &    1 &     0 &  \frac{1}{2} &   \frac{1}{2} &     0 &    0 &  \frac{1}{2} &     0 &    0 \end{bmatrix}^T \\
\vomega_3 & = \begin{bmatrix}  \frac{1}{2} &    1 &   \frac{1}{2} &  \frac{1}{2} &  \frac{1}{2} &  \frac{1}{2} &  \frac{1}{2} &    0 &   0 &    0 &    0 &    0 &    0 &    0 &    0 &    0 \end{bmatrix}^T \\
\vomega_4 & = \begin{bmatrix}  \frac{1}{2} &    0 &   0 & \frac{1}{2} &  \frac{1}{2} &    0 &   0 &   1 &     0 &  \frac{1}{2} &   \frac{1}{2} &     0 &    0 &  \frac{1}{2} &     0 &    0 \end{bmatrix}^T \\
\vomega_5 & = \begin{bmatrix}   0 &   0 & \frac{1}{2} &    0 &   0 & \frac{1}{2} &  \frac{1}{2} &    0 &   1 &      0 &    0 &  \frac{1}{2} &   \frac{1}{2} &     0 &    0 &  \frac{1}{2}  \end{bmatrix}^T \\
\vomega_6 & = \begin{bmatrix}  0 &   0 &   0 &   0 &   0 &   0 &   0 &   0 &   0 &  \frac{1}{2} &   \frac{1}{2} &   \frac{1}{2} &   \frac{1}{2} &   \frac{1}{2} &     1 &    \frac{1}{2}  \end{bmatrix}^T \\
\vomega_7 & = \begin{bmatrix}  \frac{1}{2} &    1 &   \frac{1}{2} &  \frac{1}{2} &  \frac{1}{2} &  \frac{1}{2} &  \frac{1}{2} &    0 &   0 &  \frac{1}{2} &   \frac{1}{2} &   \frac{1}{2} &   \frac{1}{2} &   \frac{1}{2} &     1 &    \frac{1}{2}  \end{bmatrix}^T \\
\vomega_8 & = \begin{bmatrix}    0 &   0 &   0 &   0 &   0 &   0 &   0 &   0 &   0 &    0 &    0 &    0 &    0 &    0 &    0 &    0 \end{bmatrix}^T.
\end{align*}
The sizes of these orbits are 8, 8, 8, 8, 8, 8, 2, and 32, respectively, making for a total of 82 pseudocodewords.  

Modulated versions of $\vomega_2, \vomega_3, \dots, \vomega_8$ were found in a fashion similar to that described above for $\vomega_1$.  The inradii of the corresponding approximation cones were 0.49, 0.43, 0.43, 0.43, 0.43, 0.32, and 0.50, respectively -- the apparent multiplicity of some of these inradii is due to the rotational symmetry of the normal graph.  It is interesting to note that the modulated version of $\vomega_8$, which is simply the all-zeros codeword, is not proportional to the vector of all ones, which is the vector used to transmit the all-zeros codeword under BPSK modulation.  To two decimal places, we have
\begin{align*}
t(\vomega_8) & = \footnotesize{\begin{bmatrix} 0.27 & 0.04 & 0.28 & 0.28 & 0.31 & 0.31 & 0.27 & 0.04 & 0.04 & 0.27 & 0.31 & 0.31 & 0.27 & 0.27 & 0.04 & 0.28\\ \end{bmatrix}}^T.
\end{align*}
With respect to the approximation cone used to produce $t(\vomega_8)$, the radius of the largest sphere centered at $t(\vomega_8)$ and contained within the approximation cone was 0.50.  The radius of the largest sphere centered at the unit vector in the direction of the BPSK vector and contained within the same approximation cone was only 0.43, which suggests that the use of $t(\vomega_8)$ will result in a lower block error rate than traditional BPSK modulation.  This suspicion is confirmed by the simulation results displayed in Figure~\ref{fig:stopSignSigma}.

\begin{figure}
\centering
\includegraphics[scale = .7]{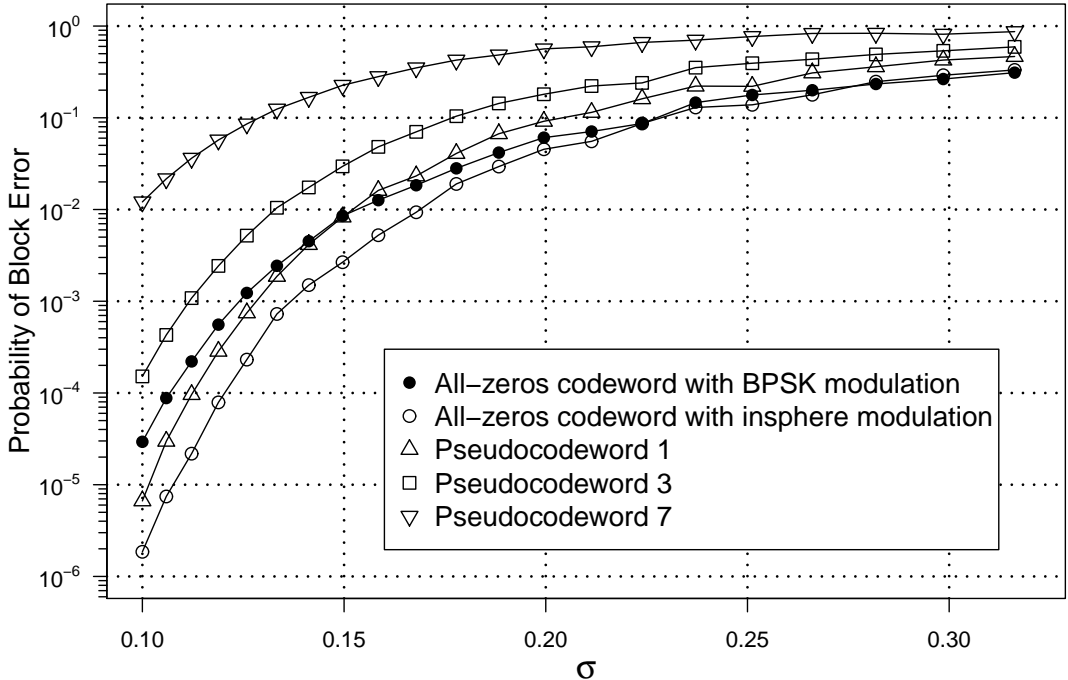}
\caption{A comparison of block error rates when transmitting the all-zeros codeword, $\vomega_1, \vomega_3$, and $\vomega_7$ with insphere modulation.  Also included is the block error rate of the all-zeros codeword under traditional BPSK modulation. The horizontal axis measures the standard deviation of the additive white Gaussian noise. }
\label{fig:stopSignSigma}
\end{figure}

Figure~\ref{fig:stopSignSigma} reveals some interesting facts.  Most basically, pseudocodewords can be reliably recovered when insphere modulation is used for transmission.  Second, pseudocodewords in different $\mc C$-symmetry orbits display a wide range of error-correcting ability:  $\vomega_7$ has terrible block error rates, while $\vomega_1$'s error rates are very close to the error rates provided by BPSK modulation for codewords.  Finally,  for low levels of noise we note that using insphere modulation for codewords results in slightly lower error rates than BPSK modulation.   As we will see in Section~\ref{sec:performance}, it is often the case where codewords themselves can benefit by exchanging BPSK modulation for insphere modulation.

\section{Performance Analysis}\label{sec:performance}

In Section~\ref{sec:csym} we saw how $\mc C$-symmetry can be leveraged to produce a codebook that is larger than the original binary code but still has the same length (i.e., the same number of real dimensions needed for its signal constellation), and a specific example was given in Section~\ref{sec:example} in which a cycle code of size $32$ could be extended to a codebook consisting of as many as 82 different messages.  This increase in spectral efficiency must be weighed against the potential loss of error-correcting performance.   For the simulations in this section, we are concerned only with the probability of a block error.  

We begin with some notation.  Let $\mc C$ be a code presented by a parity-check matrix $H$, and let $\mc O_1, \mc O_2, \dots, \mc O_m$ denote $m$ mutually disjoint orbits of LP pseudocodewords (trivial or nontrivial) under $\mc C$-symmetry. Suppose we wish to transmit information across the AWGN channel using the codebook $\mc M = \cup_{i = 1}^m \mc O_i$.  By Theorem~\ref{thm:transmit}, we may assume that for a fixed noise level the probability of block error will be uniform across each orbit $\mc O_i$ -- we denote this probability by $p_i$.  It is of the utmost importance to observe that $p_i$ is completely independent of the other orbits comprising the codebook $\mc M$.  Indeed, under LP decoding \emph{all} points in the fundamental polytope, and thus all LP pseudocodewords, are made available as potential explanations of the received vector regardless of whether these points were included in the codebook $\mc M$.  Because of this independence, if we assume a uniform distribution on the $M = \sum_{i=1}^m |\mc O_i|$ pseudocodewords from $\mc M$ we then have that the probability of block error is $p = \frac{1}{M} \sum_{i=1}^m |\mc O_i| p_i$.

This formula for the overall probability of block error gives great flexibility to the researcher when searching through possible codebooks.  First, we choose a set of linear programming pseudocodewords $\vomega_1, \dots, \vomega_m$ from distinct $\mc C$-symmetry orbits and then compute modulated versions for each (see Sections~\ref{sec:bestVector} and~\ref{sec:insphere}).  For each pseudocodeword chosen we simulate transmission across an additive white Gaussian noise channel, taking care to ensure that all modulated vectors have unit energy and that the same level of Gaussian noise is added to every pseudocodeword.  Once simulation results are obtained, we have approximations of all $p_i$'s.  From here we can chose any collection of the $m$ mutually disjoint orbits represented by $\vomega_1, \dots, \vomega_m$ to combine into an overall codebook whose probability of error is given above.  The corresponding signal-to-noise ratios are computed according to the formula $\frac{E_b}{N_0} = \frac{1}{2\log_{2}(M)\sigma^2}$, where $M$ is the size of the overall codebook and $\sigma^2$ the variance of the Gaussian noise added to the transmitted pseudocodewords.

\subsection{The cycle code $\mc C_1$}\label{subsec:c1}

For our first round of simulations we use the cycle code $\mc C_1$ from Section~\ref{sec:example}.  In Figure~\ref{fig:stopSignSim} we compare the performance of four codebook/modulation schemes.  Two of these use only the set of 32 codewords, but differ in that one utilizes BPSK modulation while the other employs insphere modulation.  The third scheme's codebook is comprised of all 32 codewords in addition to the 48 nontrivial pseudocodewords in the $\mc C$-symmetric orbits of $\vomega_1, \vomega_2, \dots, \vomega_6$, making for an overall codebook of size 80 (all vectors in this codebook were modulated using insphere modulation).  The final combination listed is the codebook consisting only of the 48 nontrivial pseudocodewords in the $\mc C$-symmetric orbits of $\vomega_1, \vomega_2, \dots, \vomega_6$ -- no binary codewords whatsoever\footnote{The two pseudocodewords in the orbit of $\vomega_7$ were not used due to their relatively poor block error rates (see Figure~\ref{fig:stopSignSigma}).}.

\begin{figure}
\centering
\includegraphics[scale = .7]{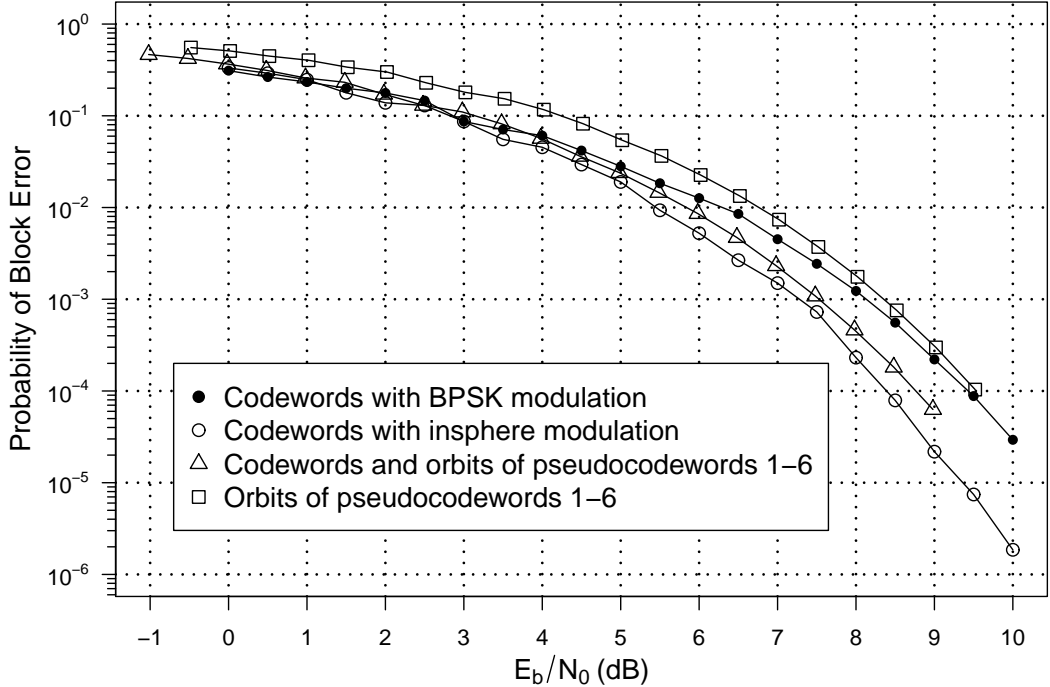}
\caption{A comparison of block error rates for four codebook/modulation pairs, all of which were derived from the cycle code $\mc C_1$ of Section~\ref{sec:example}.}
\label{fig:stopSignSim}
\end{figure}

We can make several observations from Figure~\ref{fig:stopSignSim}.  First, we see that insphere modulation enjoys a modest coding gain over BPSK modulation when it comes to transmitting codewords.  Second, we see that when pseudocodewords are broadcast in conjunction with codewords, it is possible to attain error rates better than those of codewords alone under BPSK modulation.  Moreover, these better error rates are coupled with a higher spectral efficiency:  the coding scheme consisting of codewords and the  $\mc C$-symmetric orbits of  $\vomega_1, \vomega_2, \dots, \vomega_6$ has a spectral efficiency of $\log_{2}(80)/16 = 0.3951$, whereas the spectral efficiency of the codebooks using only binary codewords is 0.3125.  The spectral efficiency of the codebook consisting only of the union of the $\mc C$-symmetry orbits of $\vomega_1, \vomega_2, \dots, \vomega_6$ is $\log_2(48)/16 = .3491$, thus landing this codebook in a bit of gray area:  when compared to codewords under BPSK modulation, it is slightly better with respect to spectral efficiency yet slightly worse with respect to block error rate.

\subsection{The cycle code $\mc C_2$}\label{subsec:c2}

A second [20, 7] cycle code $\mc C_2$ was also investigated.  The parity-check matrix chosen for $\mc C_2$ corresponds to the normal graph shown in Figure~\ref{fig:c2normalGraph}.  Using Theorem~\ref{thm:extremepointschar} and the cycle-path method of~\cite{axvigDreher} one can show that this code possesses a total of 2,376 pseudocodewords:  128 binary codewords, $\binom{8}{4} = 70$ disjoint $\mc C$-symmetry orbits consisting of $2^5 = 32$ nontrivial pseudocodewords each, and a final orbit consisting of $2^3 = 8$ nontrivial pseudocodewords.  It is therefore possible to construct a coding scheme from this code that possesses a spectral efficiency of $\log_2(2376)/20 = 0.5607$, a fair bit higher than 0.35, the spectral efficiency of the original binary code.  Many of these non-trivial pseudocodewords, however, have approximation cones with small inradii, which would drive the block error rates of a coding scheme using all 2,376 pseudocodewords to unacceptably high levels.  The task at hand, and indeed for any coding scheme we try to extend by using LP pseudocodewords, is to determine which of these non-trivial pseudocodewords enjoy a performance advantage over binary codewords under BPSK modulation.

\begin{figure}[h]
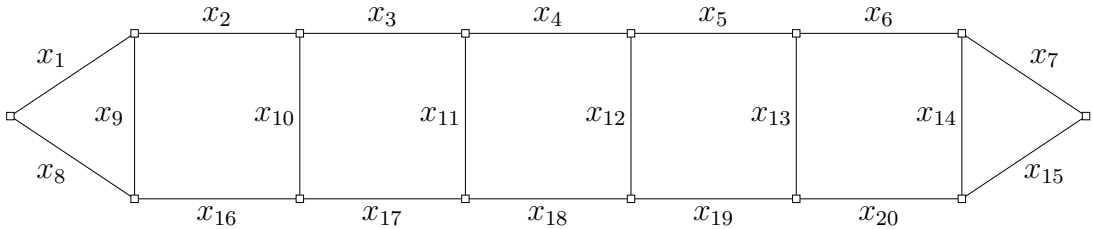

\centering
$\begin{graph}(17,4)
\unitlength=.55cm

\squarenode{f1}(2,3)[\graphnodecolour(1,1,1)]
\squarenode{f2}(5,5)[\graphnodecolour(1,1,1)]
\squarenode{f3}(9,5)[\graphnodecolour(1,1,1)]
\squarenode{f4}(13,5)[\graphnodecolour(1,1,1)]
\squarenode{f5}(17,5)[\graphnodecolour(1,1,1)]
\squarenode{f6}(21,5)[\graphnodecolour(1,1,1)]
\squarenode{f7}(25,5)[\graphnodecolour(1,1,1)]
\squarenode{f8}(28,3)[\graphnodecolour(1,1,1)]
\squarenode{f9}(5,1)[\graphnodecolour(1,1,1)]
\squarenode{f10}(9,1)[\graphnodecolour(1,1,1)]
\squarenode{f11}(13,1)[\graphnodecolour(1,1,1)]
\squarenode{f12}(17,1)[\graphnodecolour(1,1,1)]
\squarenode{f13}(21,1)[\graphnodecolour(1,1,1)]
\squarenode{f14}(25,1)[\graphnodecolour(1,1,1)]

\edge{f1}{f2}
\edge{f2}{f3}
\edge{f3}{f4}
\edge{f4}{f5}
\edge{f5}{f6}
\edge{f6}{f7}
\edge{f7}{f8}
\edge{f1}{f9}
\edge{f2}{f9}
\edge{f3}{f10}
\edge{f4}{f11}
\edge{f5}{f12}
\edge{f6}{f13}
\edge{f7}{f14}
\edge{f8}{f14}
\edge{f9}{f10}
\edge{f10}{f11}
\edge{f11}{f12}
\edge{f12}{f13}
\edge{f13}{f14}

\freetext(3,4.4){$x_1$}
\freetext(7,5.4){$x_2$}
\freetext(11,5.4){$x_3$}
\freetext(15,5.4){$x_4$}
\freetext(19,5.4){$x_5$}
\freetext(23,5.4){$x_6$}
\freetext(27,4.4){$x_7$}
\freetext(3,1.6){$x_8$}
\freetext(4.4,3.0){$x_9$}
\freetext(8.4,3.0){$x_{10}$}
\freetext(12.4,3.0){$x_{11}$}
\freetext(16.4,3.0){$x_{12}$}
\freetext(20.4,3.0){$x_{13}$}
\freetext(24.4,3.0){$x_{14}$}
\freetext(27,1.6){$x_{15}$}
\freetext(7,0.6){$x_{16}$}
\freetext(11,0.6){$x_{17}$}
\freetext(15,0.6){$x_{18}$}
\freetext(19,0.6){$x_{19}$}
\freetext(23,0.6){$x_{20}$}
\end{graph}
$
\caption{The normal graph corresponding to the [20, 7] cycle code $\mc C_2$.}
\label{fig:c2normalGraph}
\end{figure}

Many pseudocodewords were examined, and we present results on the following three for illustration:
\begin{align*}
\vomega_A & = \begin{bmatrix} \frac{1}{2}  & 1 & 1 & 1 & 1 & 1 & \frac{1}{2}  & \frac{1}{2}  & \frac{1}{2}  & 0 & 0 & 0 & 0 & \frac{1}{2}  & \frac{1}{2}  & 0 & 0 & 0 & 0 & 0 \end{bmatrix}^T \\
\vomega_B & = \begin{bmatrix} \frac{1}{2}  & 1 & 1 & 1 & \frac{1}{2} & 0 & 0  & \frac{1}{2}  & \frac{1}{2}  & 0 & 0 & \frac{1}{2} & \frac{1}{2} & 0  & 0  & 0 & 0 & 0 & \frac{1}{2} & 0 \end{bmatrix}^T \\
\vomega_C & = \begin{bmatrix} 0 & \frac{1}{2} & \frac{1}{2} &  1 & \frac{1}{2} & \frac{1}{2} & 0 & 0 & \frac{1}{2} & 0 & \frac{1}{2} & \frac{1}{2} & 0 & \frac{1}{2} & 0 & \frac{1}{2} & \frac{1}{2} & 0 & \frac{1}{2} & \frac{1}{2} \end{bmatrix}^T.
\end{align*}
Modulated versions of these pseudocodewords as well as for the all-zeros codeword were found in a manner similar to that outlined for $\mc C_1$ in Section~\ref{sec:example}, with the inradius of each approximation cone computed to within $10^{-5}$.  To two decimal places, the inradii for $\vomega_A, \vomega_B, \vomega_C$, and the all-zeros codeword were 0.46, 0.39, 0.21, and 0.45, respectively.  The radius of the largest sphere centered at the unit vector in the direction of the vector used to transmit the all-zeros codewords under BPSK modulation and contained within the approximation cone developed for the all-zeros codeword was only 0.39.  These inradii suggest that $\vomega_A$ will achieve the best error rates, followed by   codewords broadcast with insphere modulation, a near-tie between $\vomega_B$ and the all-zeros codeword under BPSK modulation, and finally $\vomega_C$.  These predictions based on inradii of approximation cones are only partly born out in simulation - see Figure~\ref{fig:box1by7sigma}.  Still, it is interesting to observe that insphere modulation of the all-zeros codeword as well as of $\vomega_A$ beat out BPSK modulation at moderate to low levels of Gaussian noise.

\begin{figure}[h]
\centering
\includegraphics[scale = .7]{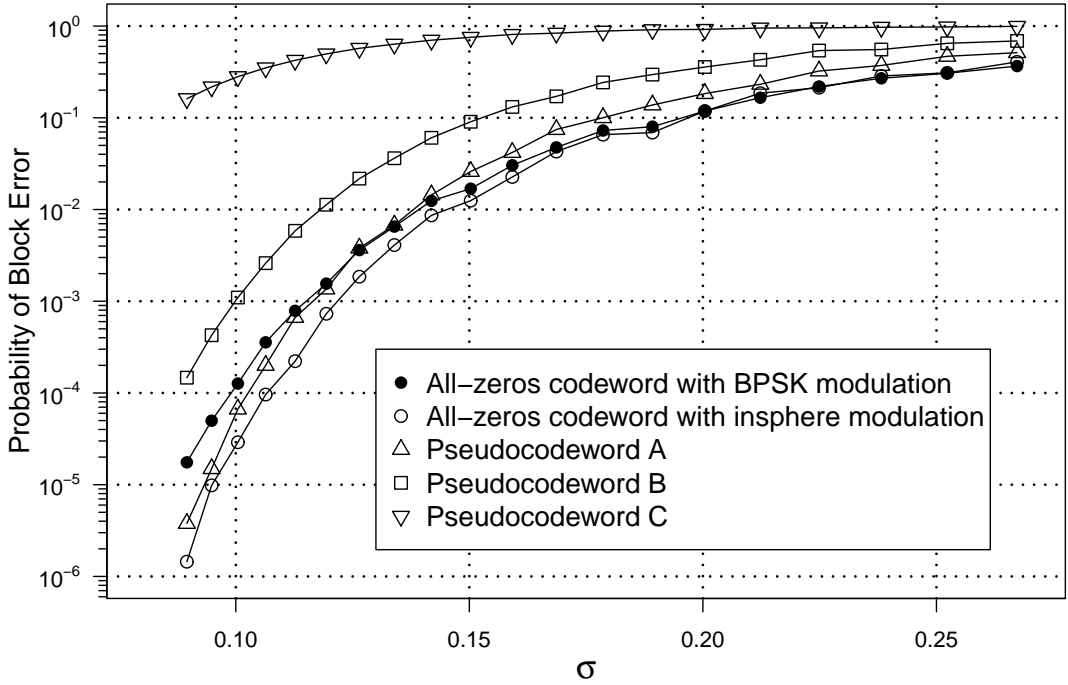}
\caption{A comparison of block error rates when transmitting the all-zeros codeword, $\vomega_A, \vomega_B$, and $\vomega_C$ with insphere modulation.  Also included is the block error rate of the all-zeros codeword under traditional BPSK modulation. The horizontal axis measures the standard deviation of the additive white Gaussian noise.}
\label{fig:box1by7sigma}
\end{figure}

We present various coding schemes based off of $\mc C_2$ and its pseudocodewords in Figure~\ref{fig:box1by7simulation}.  The spectral efficiencies of the both schemes utilizing only codewords is 0.35.  The spectral efficiency of the scheme whose codebook consists of both codewords as well as pseudocodewords in the orbit of $\vomega_A$ is $\log_2(128 + 32)/20 = 0.3661$,  while that of the scheme whose codebook consists of codewords as well as pseudocodewords in the orbits of $\vomega_A$ and $\vomega_B$ is $\log_2(128 + 32 + 32)/20 = 0.3792$.  All of the coding schemes utilizing insphere modulation have spectral efficiencies that meet or exceed that of the traditional BPSK modulation, and they do so while performing better with respect to block error rate (see Figure~\ref{fig:box1by7simulation}).  It is also interesting that the coding scheme consisting of codewords as well as pseudocodewords in the orbit of $\vomega_A$ performs \emph{slightly} better than codewords under insphere modulation while also enjoying a higher spectral efficiency.  This is encouraging, for it gives us hope that incorporating pseudocodewords into coding schema may be able to produce results that extend beyond those obtained by simply optimizing the manner in which we transmit binary codewords.  

\begin{figure}[h]
\centering
\includegraphics[scale = .7]{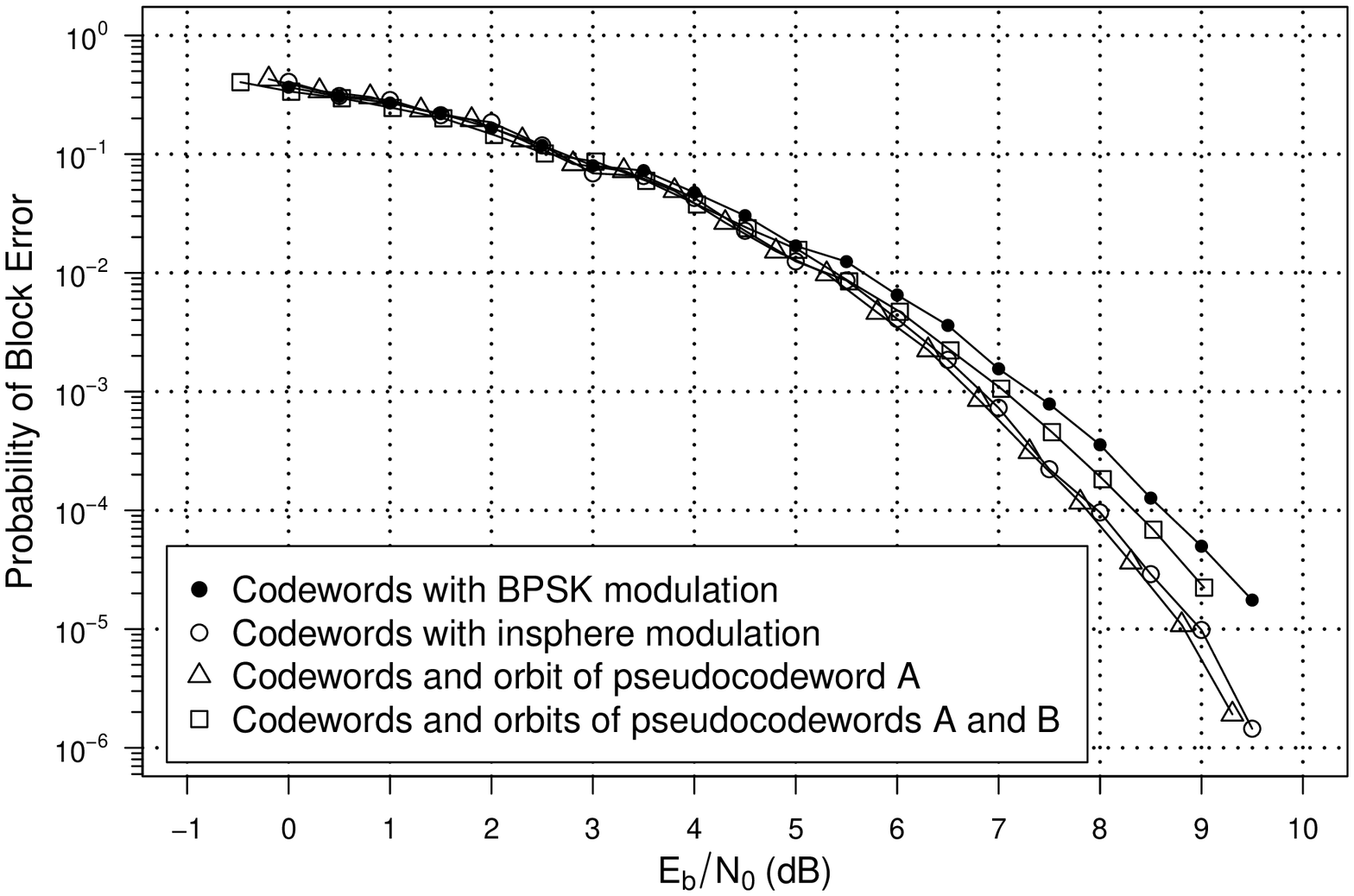}
\caption{A comparison of block error rates for four codebook/modulation pairs derived from the cycle code $\mc C_2$.}
\label{fig:box1by7simulation}
\end{figure}

\subsection{A Randomly Generated [30, 10] Code}\label{subsec:randomn30}

In order to branch out from cycle codes we randomly generated a $[30, 10]$ code $\mc C_3$ with the following parity-check matrix:
\[
\setcounter{MaxMatrixCols}{30}
H_3 = \footnotesize{\begin{bmatrix}
0&0&0&0&0&0&0&0&0&0&1&0&0&0&1&0&0&0&0&0&0&0&0&0&0&0&0&0&0&0 \\
0&0&0&0&0&0&0&0&0&0&0&0&0&1&0&0&0&0&0&0&1&0&0&0&0&0&0&0&0&0 \\
0&0&0&0&1&0&0&0&0&0&0&1&0&0&0&0&0&0&0&0&0&0&0&0&0&0&0&0&0&0 \\
1&0&0&0&0&0&0&0&0&1&0&0&0&0&0&0&0&0&0&0&0&0&0&0&0&0&0&0&0&0 \\
0&0&0&0&0&0&0&0&0&0&0&1&0&0&0&0&0&0&0&0&0&0&0&1&1&0&0&0&0&0 \\
0&0&1&0&0&0&1&0&0&0&0&0&0&1&0&0&0&0&0&0&0&0&0&0&0&0&0&0&0&0 \\
0&0&0&0&0&0&0&0&0&0&0&0&1&0&0&0&0&0&0&0&0&0&0&0&0&0&1&0&1&0 \\
0&0&0&0&0&0&0&0&0&0&0&1&0&0&0&0&1&0&0&0&0&0&0&0&0&1&0&0&0&0 \\
0&0&0&0&0&0&0&0&0&0&0&0&0&0&0&1&0&0&0&0&0&0&1&0&0&1&0&0&0&0 \\
0&0&0&0&1&0&0&0&0&0&1&0&1&0&0&0&0&0&0&0&0&0&0&0&0&0&0&0&0&0 \\
0&0&0&0&0&0&0&0&1&0&0&0&0&0&0&1&0&0&0&0&0&1&0&0&0&0&0&0&0&0 \\
0&0&0&0&0&1&0&0&0&0&0&0&0&0&0&1&0&0&0&0&0&1&0&0&0&0&0&0&0&0 \\
0&0&1&0&0&0&0&0&0&0&1&0&0&0&0&0&0&0&0&0&0&1&0&0&0&0&0&0&0&0 \\
0&1&0&0&0&0&0&0&1&0&0&0&0&0&0&0&1&0&0&0&0&0&0&0&0&0&0&0&0&0 \\
0&1&0&0&0&0&0&0&0&0&0&0&1&0&0&0&0&0&0&0&0&0&0&0&0&0&0&1&0&0 \\
0&1&0&0&0&0&0&0&0&0&0&0&0&1&0&0&0&1&1&0&0&0&0&0&0&0&0&0&0&0 \\
0&0&0&0&0&0&0&1&0&0&0&0&0&0&1&0&0&0&0&0&0&0&0&0&0&1&0&1&0&0 \\
0&0&0&1&0&0&0&1&0&0&0&0&0&0&0&0&0&0&1&0&0&0&0&0&1&0&0&0&0&0 \\
0&0&0&0&0&0&0&0&0&0&0&0&1&0&0&0&0&1&0&1&0&0&0&0&0&0&0&0&1&0 \\
0&0&0&0&0&0&1&1&0&0&0&0&0&0&0&0&0&1&0&0&0&0&0&0&1&0&0&0&0&1
\end{bmatrix}}.
\]
Pseudocodewords for this code/parity-check matrix combination were generated via a random search.  The two we highlight are given below:
\begin{align*}
\vomega_D & = \footnotesize{\begin{bmatrix}  0	& 1	& 0	& 0	& \frac{1}{3}	& 0	& 1	& \frac{2}{3}		& 0	& 0	& \frac{1}{3}		& \frac{1}{3}		& 0	& 1	& \frac{1}{3}		& \frac{1}{3}		& 1	& 0	& 0	& 1	& 1	& \frac{1}{3}		& 1	& 1	& \frac{2}{3}	& \frac{2}{3}		& 1	& 1	& 1	& 1\end{bmatrix} }^T\\
\vomega_E & = \footnotesize{\begin{bmatrix} 1	& \frac{4}{5}	& 0	& 1	& 0	& 1	& 0	& \frac{3}{5}	& 1	& 1	& 1	& 0	& 1	& 0	& 1	& 0	& \frac{1}{5}	& \frac{3}{5}	& \frac{3}{5}	& \frac{3}{5}	& 0	& 1	&  \frac{1}{5}		& 1	& 1	&  \frac{1}{5}		& 0	&  \frac{1}{5}		& 1	& 1\end{bmatrix}}^T.
\end{align*}
Since neither $\vomega_D$ nor $\vomega_E$ contains any fractional entries equal to $\frac{1}{2}$, each has a $\mc C$-symmetry orbit of the same size as the binary code itself.  The codebooks that included these nontrivial pseudocodewords, however, all had much higher error rates than the binary codewords. Figure~\ref{fig:randomn30sigma} shows the performance of $\vomega_D$ and $\vomega_E$ as compared to the performance of the all-zeros codeword under both BPSK and insphere modulation.  The inradii of the approximation cones used to produce transmittable versions of $\vomega_D$ and $\vomega_E$ were 0.24 and 0.20, respectively.  The corresponding inradii for the all-zeros codeword under insphere modulation was, by comparison, 0.43.  The gap between these numbers is likely the reason for the marked difference in block error rates observed in Figure~\ref{fig:randomn30sigma}.  Since none of the nontrivial pseudocodewords we found came close to matching the performance of codewords, we were unable to increase the spectral efficiency of this coding scheme without greatly sacrificing performance - see Figure~\ref{fig:randomn30sim}.

\begin{figure}[b!]
\centering
\includegraphics[scale = .7]{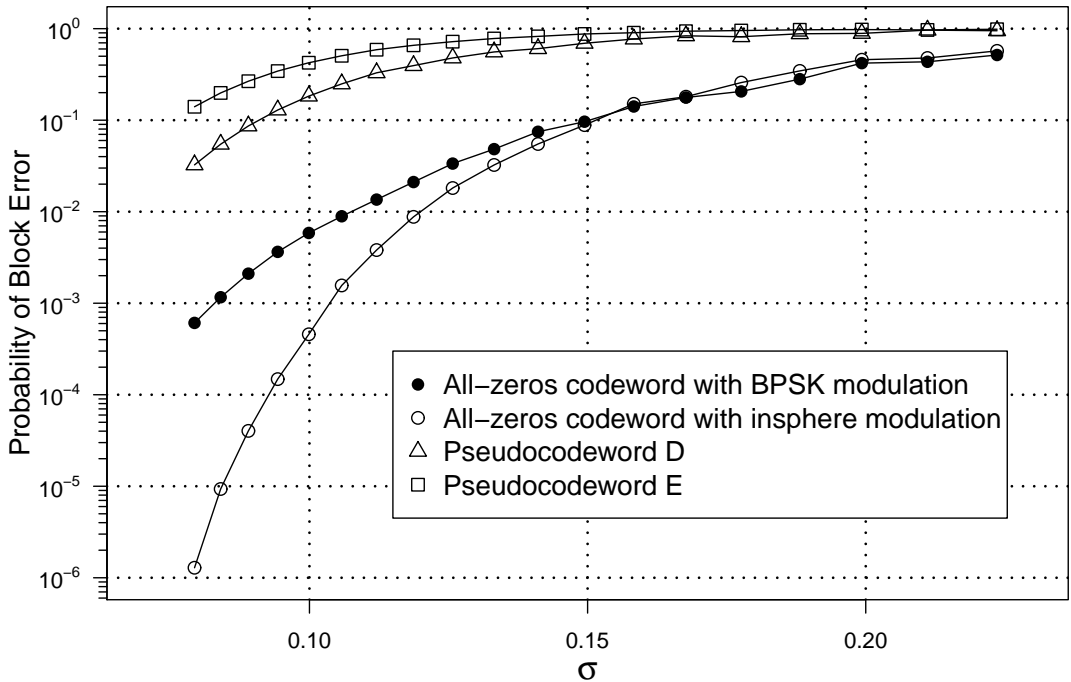}
\caption{A comparison of block error rates for the randomly generated [30, 10] code when transmitting the all-zeros codeword and two nontrivial pseudocodewords $\vomega_D$ and $\vomega_E$ with insphere modulation.  Also included is the block error rate of the all-zeros codeword under traditional BPSK modulation. The horizontal axis measures the standard deviation of the additive white Gaussian noise. }
\label{fig:randomn30sigma}
\end{figure}

\begin{figure}[h]
\centering
\includegraphics[scale = .7]{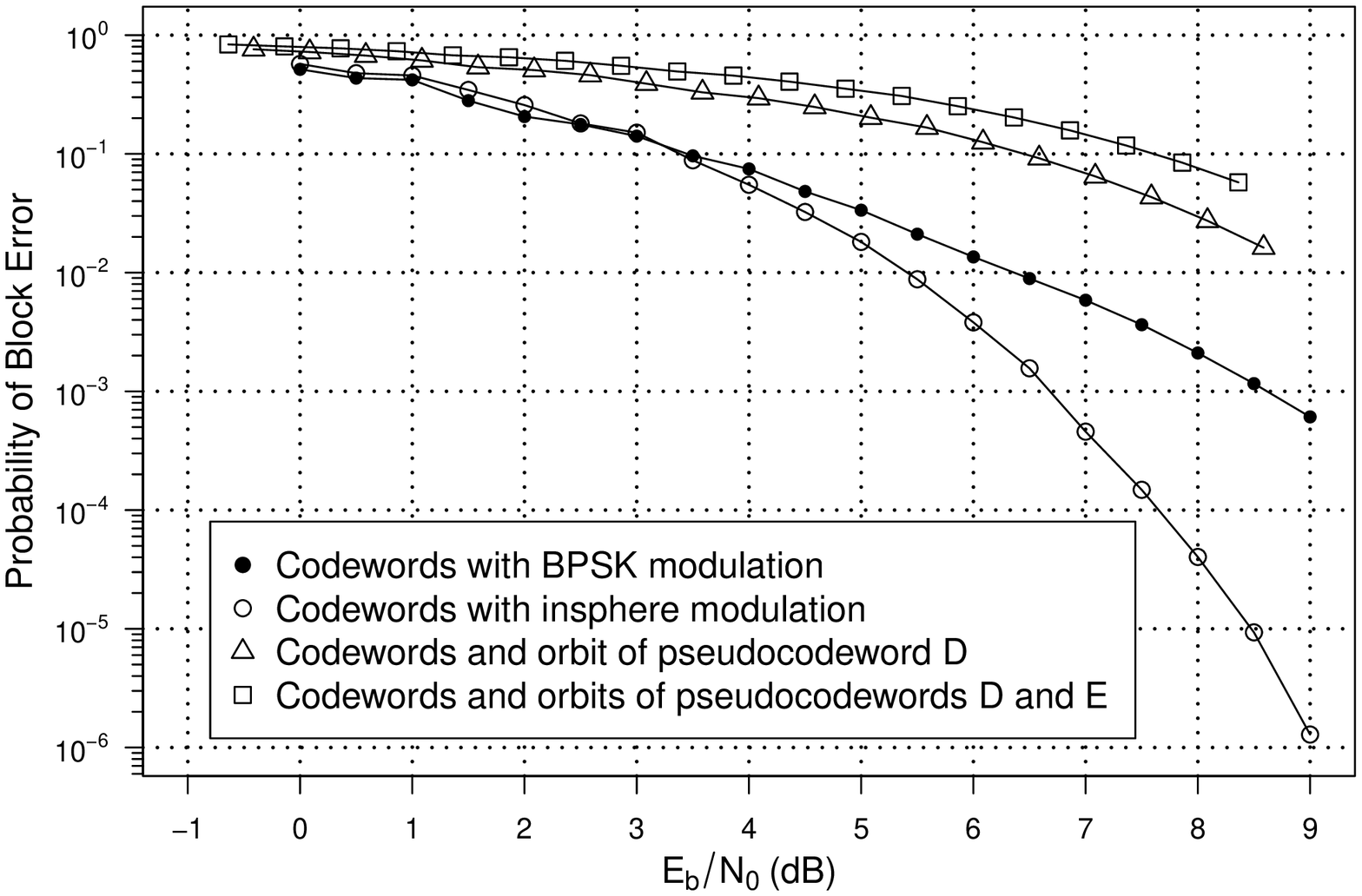}
\caption{A comparison of block error rates for four codebook/modulation pairs derived from the code $\mc C_3$.  Insphere modulation is used for all elements of a codebook unless explicitly noted otherwise.}
\label{fig:randomn30sim}
\end{figure}

As mentioned before, the inradius of the approximation cone used to produce a transmittable version of the all-zeros codeword was 0.43.  The largest sphere contained within this same cone but centered at the unit vector in the direction of the BPSK vector was only 0.26.  This suggests that insphere modulation of codewords should display lower block error rates than those of codewords under BPSK modulation.  As we see in Figure~\ref{fig:randomn30sigma}, this is indeed the case.  A potential explanation of this improvement in performance is that insphere modulation gives a method of identifying individual bit locations that are particularly susceptible to errors and ``boosting'' them, while in turn lowering the energy used to transmit other bits to balance things out.  For instance, we note that $\mc C_3$ has a codeword of Hamming weight 2 whose support lies only in the first and tenth positions. The unit vector used to transmit the all-zeros vector under insphere modulation assigns values of 0.3048 and 0.3065 to the first and tenth positions, respectively.  These values are significantly larger than $1/\sqrt{30} = 0.1826$, which is the value that the unit vector in the direction of the BPSK vector assigns uniformly to all positions.

%
%

\subsection{A Randomly Generated [30, 17] Code}

In the past three examples, insphere modulation applied to codewards has always outperformed binary phase shift keying at higher signal-to-noise ratios.  In this section, we give a cautionary example to demonstrate that this is not always the case.

Consider the $[30, 17]$ code $\mc C_4$ presented by the following parity-check matrix:
\[
\setcounter{MaxMatrixCols}{30}
H_4 = \footnotesize{\begin{bmatrix}
1&	0&	1&	0&	0&	0&	0&	1&	0&	0&	0&	1&	0&	0&	0&	0&	0&	0&	0&	0&	0&	0&	1&	0&	0&	0&	0&	0&	0&	1\\
0&	0&	0&	0&	0&	0&	0&	0&	0&	1&	0&	0&	1&	0&	0&	1&	0&	1&	0&	0&	1&	0&	0&	0&	1&	0&	0&	0&	0&	0\\
0&	0&	0&	0&	0&	1&	1&	0&	0&	0&	1&	0&	0&	0&	0&	0&	0&	0&	0&	0&	0&	1&	0&	0&	0&	0&	0&	1&	1&	0\\
0&	0&	0&	1&	0&	0&	0&	0&	1&	0&	0&	0&	0&	1&	1&	0&	0&	0&	0&	1&	0&	0&	0&	1&	0&	0&	0&	0&	0&	0\\
0&	1&	0&	0&	1&	0&	0&	0&	0&	0&	0&	0&	0&	0&	0&	0&	1&	0&	1&	0&	0&	0&	0&	0&	0&	1&	1&	0&	0&	0\\
1&	0&	0&	0&	0&	0&	0&	0&	0&	0&	0&	0&	1&	0&	0&	1&	0&	0&	0&	0&	0&	0&	0&	1&	0&	0&	0&	0&	1&	1\\
0&	0&	0&	1&	0&	0&	0&	0&	0&	0&	0&	0&	0&	1&	0&	0&	1&	0&	0&	0&	1&	1&	0&	0&	0&	0&	0&	1&	0&	0\\
0&	1&	0&	0&	0&	0&	0&	0&	1&	0&	1&	1&	0&	0&	0&	0&	0&	0&	1&	0&	0&	0&	0&	0&	0&	1&	0&	0&	0&	0\\
0&	0&	1&	0&	1&	1&	0&	1&	0&	0&	0&	0&	0&	0&	1&	0&	0&	0&	0&	0&	0&	0&	0&	0&	0&	0&	1&	0&	0&	0\\
0&	0&	0&	0&	0&	0&	1&	0&	1&	0&	0&	0&	0&	0&	0&	1&	0&	0&	0&	0&	0&	0&	1&	0&	1&	1&	0&	0&	0&	0\\
0&	0&	0&	0&	1&	0&	0&	0&	0&	1&	0&	0&	1&	1&	0&	0&	0&	0&	0&	0&	0&	1&	0&	0&	0&	0&	0&	0&	1&	0\\
1&	0&	1&	0&	0&	0&	0&	0&	0&	0&	0&	0&	0&	0&	0&	0&	1&	0&	1&	1&	1&	0&	0&	0&	0&	0&	0&	0&	0&	0\\
0&	0&	0&	0&	0&	0&	0&	1&	0&	0&	0&	0&	0&	0&	1&	0&	0&	1&	0&	0&	0&	0&	0&	1&	0&	0&	1&	0&	0&	1
\end{bmatrix}}.
\]
All pseudocodewords tested for this code performed abysmally - the best was still giving block error rates above 0.1 at 8.5 decibels.  Moreover, as we can see in Figure~\ref{fig:randomn30k17sim} the performance of insphere modulation and BPSK modulation were virtually identical.  Thus, no real gains in either spectral efficiency or block error rate can be gained by using insphere modulation in conjunction with this code.

\begin{figure}[h]
\centering
\includegraphics[scale = .7]{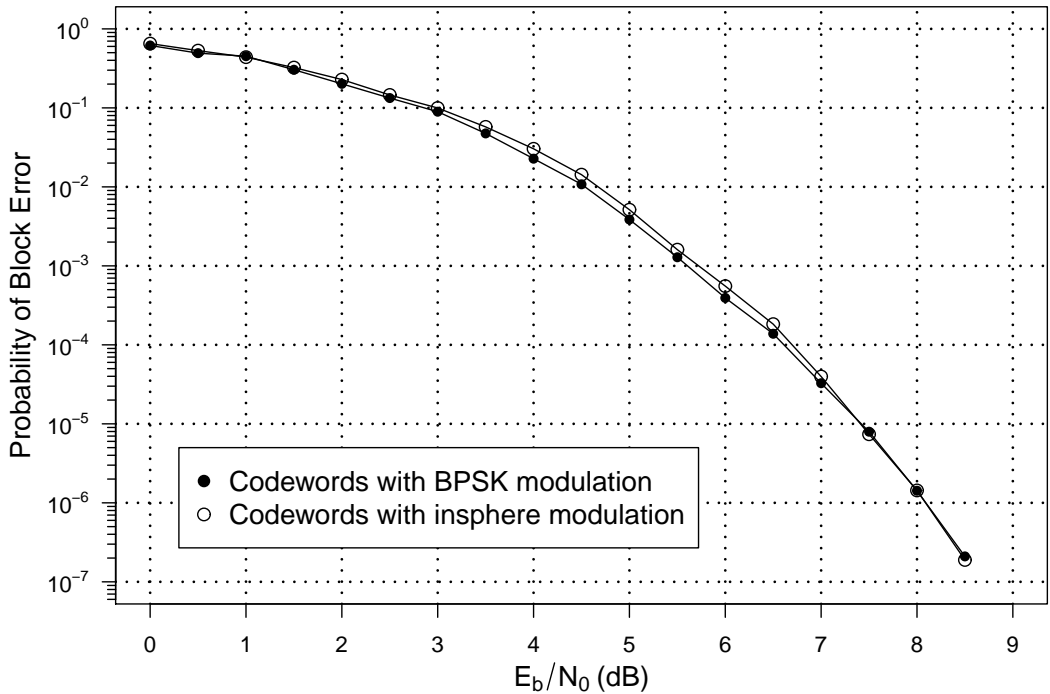}
\caption{A comparison of block error rates for the randomly generated $[30, 17]$ code under both BPSK and insphere modulation.}
\label{fig:randomn30k17sim}
\end{figure}

\section{Conclusion}\label{sec:conclusion}

In this paper we have defined insphere modulation, a modulation scheme that is suitable for transmitting nontrivial linear programming pseudocodewords in addition to binary codewords.  In doing so, we have developed a method for approximating the insphere of a polyhedral cone to any degree of accuracy.  Simulation results confirm that certain codebooks utilizing a combination of nontrivial LP pseudocodewords in conjunction with codewords give modest gains in spectral efficiency without sacrificing error correcting performance. While originally introduced as a means to transmit nontrivial pseudocodewords, insphere modulation can in some instances be used to lower the block error rates of binary codewords relative to their performance under traditional BPSK modulation.

It is the author's belief that the characterization of linear programming pseudocodewords for cycle codes of~\cite{axvigDreher} makes it easier for the researcher to search through the vast number of pseudocodewords and possibly stumble upon some that can be beneficial.  A broader characterization of LP pseudocodewords for general binary codes, should it exist, may prove  useful in determining the overall efficacy of using hybrid codebooks comprised of both binary codewords and nontrivial pseudocodewords.

\bibliographystyle{plain}
\bibliography{refs_lp_pcw_incorporation}

\end{document}